\RequirePackage{etoolbox}
\csdef{input@path}{%
}%

\documentclass[ba]{imsart}
\usepackage{amsmath}
\usepackage{graphicx, psfrag, epsf}
\usepackage{enumerate}
\usepackage{natbib}
\usepackage{url} 

\usepackage{amssymb, amsfonts, amsthm, bm}
\usepackage{mathtools, mathrsfs}
\usepackage{relsize} 
\usepackage{hyperref}

\newtheorem{definition}{Definition}
\newtheorem{theorem}{Theorem}
\newtheorem{lemma}{Lemma}

\newtheorem{remark}{Remark}
\newtheorem{example}{Example}

\def\beginmat{ \left( \begin{array} }
\def\endmat{ \end{array} \right) }

\def\log{{\rm log}}
\def\tr{{\rm tr}}

\DeclareMathOperator{\E}{E}
\DeclareMathOperator{\Var}{Var}
\DeclareMathOperator*{\argmax}{argmax} 
\DeclareMathOperator{\Length}{length}

 \newcommand{\noop}[1]{}

\usepackage{color,soul}
\soulregister\cite7
\soulregister\citep7
\soulregister\eqref7
\soulregister\pageref7

\startlocaldefs
\numberwithin{equation}{section}
\theoremstyle{plain}
\newtheorem{thm}{Theorem}[section]
\endlocaldefs

\begin{document}

\begin{frontmatter}
\title{Jointly Robust Prior for Gaussian Stochastic Process in Emulation, Calibration and Variable Selection\thanksref{T1}}
\runtitle{Jointly Robust Prior}

\begin{aug}
\author{\fnms{Mengyang} \snm{Gu}\thanksref{addr1}\ead[label=e1]{mgu6@jhu.edu}}

\runauthor{Mengyang Gu}

\address[addr1]{Department of Applied Mathematics and Statistics, Johns Hopkins University, Baltimore, MD 
    \printead{e1} 
}

%

\end{aug}
\begin{abstract}
 Gaussian stochastic process (GaSP) has been widely used in two fundamental problems in uncertainty quantification, namely the emulation and calibration of mathematical models. Some objective priors,  such as the reference prior, are studied in the context of emulating (approximating) computationally expensive mathematical models.  In this work, we introduce a new class of priors, called the jointly robust prior, for both the emulation and calibration.
This prior is designed to maintain various advantages from the reference prior. In emulation, the jointly robust prior has an appropriate tail decay rate as the reference prior, and is computationally simpler than the reference prior in parameter estimation. Moreover, the marginal posterior mode estimation with the jointly robust prior can separate the influential and inert inputs in mathematical models, while the reference prior does not have this property. We establish the posterior propriety for a large class of priors in calibration, including the reference prior and jointly robust prior in general scenarios, but the jointly robust prior is preferred because the calibrated mathematical model typically predicts the reality well. The jointly robust prior is used as the default prior in two new R packages, called ``RobustGaSP" and ``RobustCalibration",  available on CRAN for emulation and calibration, respectively. 
 \end{abstract}

 
%



\end{frontmatter}

\section{Introduction}
 
 
 
A central part of the modern uncertainty quantification (UQ) is to describe the natural and social phenomena by a system of mathematical models or equations. Some mathematical models are implemented as computer code in an effort to reproduce the behavior of complicated processes in science and engineering. These mathematical models are called {computer models} or {simulators}, which map a set of inputs such as initial conditions and model parameters to a real valued output. 

Many computer models are prohibitively slow, and it is thus vital to develop a fast statistical surrogate to emulate (approximate) the outcomes of the computer models, based on the runs at a set of pre-specified design inputs. This problem is often referred as the emulation problem. Another fundamental problem in UQ is called {the inverse problem} or {calibration}, where the field data are used to estimate the unobservable calibration parameters in the mathematical model. As the mathematical model can be imprecise to describe the reality, it is usual to address the misspecification by a discrepancy function. Emulation and calibration are the main focus in many recent studies in UQ  \citep{bayarri2007framework, higdon2008computer, Bayarri09,   liu2009modularization, conti2010bayesian}. 

%
%

A Gaussian stochastic process (GaSP) is prevalent for emulating expensive computer model \citep{sacks1989design,bastos2009diagnostics} for several reasons. First of all, many computer models are deterministic, or close to being deterministic, and thus the emulator is often required to be an interpolator, meaning that the predictions by the emulator are equal to the outputs at the design inputs. The GaSP emulator is an interpolator, and can easily be adapted to emulate the stochastic computer model outputs by adding a noise. Second, the number of runs of the computer model used to construct a GaSP emulator is often relatively small, which is roughly $n\approx 10p$ by the ``folklore" notion, where $p$ is the dimension of the inputs. Third, the GaSP emulator has an internal assessment of the accuracy in prediction, which allows  the uncertainty to propagate through the emulator. The GaSP is also widely used to model the discrepancy function in calibration \citep{kennedy2001bayesian, bayarri2007framework}, as combining the calibrated computer model and discrepancy function can typically improve the predictive accuracy than the prediction using the computer model alone.



The GaSP model used in emulation and calibration is rather different than the one in modeling spatially correlated data. The key difference is that the input space of the computer model usually has multiple dimensions and completely different scales.  The isotropic assumption is thus too restrictive. Instead, for any $\mathbf x_a, \mathbf x_b \in \mathcal X$ with $p_x$ dimensions, one often assumes a product correlation (\cite{sacks1989design})
\begin{equation}
c(\mathbf x_a, \mathbf x_b)=\prod^{p_x}_{l=1}c_l( x_{al},  x_{bl}), 
\label{equ:correlation}
\end{equation}
where each $c_l$ is a one-dimensional isotropic correlation function for the $l$th coordinate of the input, each typically having an unknown range parameter $\gamma_l$ and fixed roughness parameter $\alpha_l$,  $l=1,...,{p_x}$.  This choice of the correlation will be used herein due to its flexibility in modeling correlation and tractability in computation. 

The performance of a GaSP model in emulation and calibration depends critically on the parameter estimation of the GaSP model. For the emulation problem, it's been recognized in many studies that some routinely used methods, such as the maximum likelihood estimator (MLE), produce unstable estimates of the correlation parameters \citep{oakley1999bayesian,lopes2011development}. The instability in parameter estimation results in a great loss of the predictive accuracy, as the covariance matrix is estimated to be near-singular or near-diagonal. This problem is partly overcome by the use of the reference prior \citep{berger2001objective,paulo2005default}, where the marginal posterior mode estimation under certain parameterizations eliminates these two unwelcome scenarios (\cite{Gu2018robustness}).

Other than the reference prior, many proper and improper priors were previously studied for the GaSP model in emulation and calibration,  often with a product form with various parameterizations, including the inverse range parameter $\beta_l=1/\gamma_l$, natural logarithm of the inverse range parameter $\xi_l=\log(\beta_l)$, and correlation parameter $\rho_l=1/\exp(\beta_l)$,  $l=1,..,{p_x}$.   For example, $\pi(\beta_l)\propto 1/\beta_l$ was utilized in \cite{kennedy2001bayesian} and $\pi(\beta_l)\propto{1}/{(1+\beta^2_l)}$ was assumed in \cite{conti2010bayesian}. An independent beta prior for $\rho_l$ is utilized in \cite{higdon2008computer}, and the spike and Slab prior for the same parameterization is used in \cite{savitsky2011variable}. Though eliciting the prior information has been discussed in the literature (\cite{oakley2002eliciting}), it is rather hard to faithfully transform subjective prior knowledge to the GaSP model with the product correlation function in (\ref{equ:correlation}). 

In this work,  we propose a new class of priors, called the jointly robust (JR) prior, for both the emulation problem and calibration problem.  This prior maintains most of the  advantages of the reference prior in emulation, and it has a closed-form normalizing constant, moments and derivatives. In comparison, although the computational operations of the reference prior is normally acceptable, the derivative of the reference prior is more computationally expensive.  In practice, the numerical derivatives of the reference prior are often used for the marginal posterior mode estimation, which slows down the computation. Moreover, the prior moments and the normalizing constant of range parameters by the reference prior are unknown and even hard to compute, because of the near-singular correlation matrix when all range parameters are large.   



In the calibration problem,  we establish  the posterior propriety for the calibration problem of a wide class of priors, including the reference prior and JR prior in general scenarios. The identifiability problem of the calibration parameters was found in many previous studies (\cite{arendt2012quantification,tuo2015efficient}), partly due to the large correlation estimated by the data (\cite{gu2017improved}).  Though the posteriors of the reference prior and JR prior are shown to be proper for the calibration problem in this work, the density of the JR prior has slightly larger slope than the density of the reference prior when the range parameters in the covariance function get large, preventing the correlation from being estimated to be too large. Numerical results of the advantages of using the JR prior against the reference prior in calibration will be discussed.  Two R packages, called ``RobustGaSP" and ``RobustCalibration", are developed for the emulation and calibration problems, and the jointly robust prior is used as the default choice in both packages (\cite{gu2018robustgasp,gu2018robustcalibrationpackage}).

Furthermore, another advantage of the JR prior is that it can identify the inert inputs efficiently through the marginal posterior mode of a full model, whereas the mode with the reference prior does not have this feature. The inert inputs are the ones that barely affect the outputs of the computer model. Having inert inputs is a fairly common scenario with computer models. E.g. In TITAN2D computer model for simulating volcanic eruption \citep{Bayarri09},  the internal friction angle has a negligible effect on the output. In emulation, having an inert input can sometimes result in worse prediction than simply omitting them and in calibration, one may hope to spend more efforts in calibrating the influential inputs than the inert inputs. The full Bayesian variable selection of the inputs in a computer model is often prohibitively slow, as each evaluation of the likelihood is computationally expensive,  whereas the marginal posterior mode by the JR prior  is much faster for identifying the inert inputs, discussed in Section \ref{subsec:JR_selection}.

Compared to other frequently used priors other than the reference prior, the new class of priors studied in this work is not a product of marginal priors of the range parameter or its transformation. The advantage is that the marginal posterior mode estimation with the new prior is both robust and useful in identifying the inert inputs in the mathematical model. Using a product of marginal priors of the parameters in the covariance function may not achieve both properties at the same time.

The paper is organized as follows. In Section~\ref{sec:GaSP}, we review the GaSP model in emulation and calibration, exploring the benefit of the reference prior in emulation that were not noticed before. A general theorem about the posterior propriety is also derived in the calibration setting.  In Section~\ref{sec:JR_prior}, we introduce the JR prior, and compare with the reference prior in calibration and emulation. The variable selection problem is introduced in Section~\ref{sec:variable_selection}. The numerical studies of using the JR prior for emulation, variable selection and calibration will be discussed in Section~\ref{sec:numerical}.  We conclude the paper in Section~\ref{sec:conclusion}.

\section{Gaussian stochastic process model} 
\label{sec:GaSP}
In this section, we first shortly introduce the GaSP model in Section~\ref{subsec:background}. The model will be extended for emulation and calibration  in Section \ref{subsec:GP_emulator} and Section \ref{subsec:GP_calibration}, respectively.  The posterior propriety will also be studied in the calibration problem in Section \ref{subsec:GP_calibration}.



\subsection{Background}
\label{subsec:background}
To begin with, consider a stationary Gaussian stochastic process $y(\cdot)\in \mathbb R$ on a ${p_x}$-dimensional input space $\mathcal X$, 
 \begin{equation}
 y(\cdot)\sim \mbox{GaSP}(\mu(\cdot), \, \sigma^2 c(\cdot, \cdot))
 \label{equ:GaSP}
 \end{equation}
where $\mu(\cdot)$ and $\sigma^2 c(\cdot, \cdot)$ are the mean and covariance functions, respectively. Any marginal distribution $(y(\mathbf x_1),...,y(\mathbf x_n))^T$ follows a multivariate distribution,
\[(y(\mathbf x_1),...,y(\mathbf x_n))^T \sim \mbox{MN}(\bm \mu, \sigma^2 \mathbf R ), \]
where $\bm \mu=(\mu(\mathbf x_1),...,\mu(\mathbf x_n))^T$ is an $n$-dimensional vector of the mean, and $\sigma^2 \mathbf R$ is an $n\times n$ covariance matrix with the $(i,\,j)$ entry being $\sigma^2c(\cdot,\,\cdot)$, where $c(\cdot, \cdot)$ is a correlation function.

 The mean function for any input $\mathbf x \in \mathcal X$ is typically  modeled via the regression
\begin{equation}
\mu(\mathbf x)=  \mathbf h(\mathbf x) \bm \theta_{m}=\sum^{q}_{t=1}  h_t(\mathbf x)  \theta_{mt},
\label{equ:mean}
\end{equation}
  where $\mathbf h(\mathbf x)= (h_1(\mathbf x),...,h_q(\mathbf x))$ is a row vector of the mean basis functions and $ \theta_{mt}$ is the unknown regression parameter of the basis function $h_t(\cdot)$, for $t=1,...,q$, with $q$ being the number of the mean basis specified in the model. 
  

The product correlation function in (\ref{equ:correlation}) is assumed and thus the correlation matrix is $\mathbf R=\mathbf R_1 \circ \mathbf R_2 \circ ... \circ \mathbf R_{p_x}$,  where $\circ$ is the Hadamard product.  The $(i,j)$ entry of $\mathbf R_l$ is parameterized by $c_l(\cdot, \cdot)$, a one dimensional correlation function for the $l$th coordinate of the input, $l=1,...,{p_x}$. We focus on two classes of widely used correlation functions: the power exponential correlation and Mat{\'e}rn correlation. Define $d_l=|x_{al}-x_{bl}|$ for  any $\mathbf x_a, \mathbf x_b \in \mathcal X$. The power exponential correlation has the form
\begin{equation}
c_l(d_l)=\exp\left\{- \left(\frac{d_l}{\gamma_l} \right)^{\alpha_l}\right\},
\label{equ:pow_exp}
\end{equation}
where $\gamma_l$ is an unknown nonnegative range parameter to be estimated and $\alpha_l \in(0,2]$ is a fixed roughness parameter, often chosen to be a value close to 2 to avoid the numerical problem when $\alpha_l=2$ \citep{Bayarri09,Gu2016PPGaSP}. 

 The Mat{\'e}rn correlation has the following form 
 \begin{equation}
 c_l(d_l)=\frac{1}{2^{\alpha_l-1}\Gamma(\alpha_l)}\left(\frac{d_l}{\gamma_l} \right)^{\alpha_l} \mathcal K_{\alpha_l} \left(\frac{d_l}{\gamma_l} \right),
 \label{equ:matern}
 \end{equation}
 where  $\Gamma(\cdot)$ is the gamma function, $\mathcal{K}_{\alpha_l}(\cdot)$ is the modified Bessel function of the second kind with the range parameter and roughness parameter  being $\gamma_l$ and $\alpha_l$, respectively. The Mat{\'e}rn correlation has a closed-form expression when $\alpha_l=\frac{2k_l+1}{2}$ with $k_l\in \mathbb N$, and becomes the exponential correlation and Gaussian correlation, when $k_l=0$ and $k_l\to \infty$, respectively. Though we focus on these two classes of correlation functions, the results are applicable to other different correlation functions shown in \cite{Gu2018robustness}. 

\subsection{GaSP emulator and the reference prior}
\label{subsec:GP_emulator}
 The goal of emulation is to predict and assess the uncertainty on the real-valued output of a computationally expensive computer model, denoted as $f^M(\cdot)$, based on a finite number of chosen inputs, often selected to fill the input domain $\mathcal X$, e.g. the Latin Hypercube Design \citep{sacks1989design,santner2003design}. Let us model the unknown function $f^M(\cdot)$ via a GaSP defined in (\ref{equ:GaSP}). Denote the outputs of the computer model $\mathbf f^M=(f^M(\mathbf x_1),...f^M(\mathbf x_n))^T$  at $n$ chosen inputs $\{\mathbf x_1,...,\mathbf x_n\}$. Conditional on $ \mathbf f^M$, the GaSP emulator is to predict and quantify the uncertainty of the output at $\mathbf x^*$  by the predictive distribution of $f^M(\mathbf x^*)$.


The GaSP emulator typically consists of the mean parameters, variance parameter and range parameters, denoted as $(\bm{\theta}_m ,{\sigma}^2,\bm {\gamma})$. The reference prior for the GaSP model with the product correlation was developed in \cite{paulo2005default} and the form is given by
\begin{equation}
 {\pi }^R(\bm{\theta}_m ,{\sigma}^2,\bm {\gamma} ) \propto  \frac{\pi^R(\bm \gamma)}{{\sigma }^{2}} \,,
 \label{equ:refprior}
 \end{equation}
with ${\pi ^R}(\bm{\gamma} )\propto |{\mathbf I^{*}}(\bm{\gamma} ){|^{1/2}}$, where $\mathbf I^{*}(\cdot)$ is the expected Fisher information matrix as below
\begin{equation}
 \label{equ:efi}
 {\mathbf I^*}({\bm \gamma} ) = {\left( {\begin{array}{*{20}{c}}
   {n - q} & {\tr({\mathbf W_1})} & {\tr({\mathbf W_2})} & {...} & {\tr({\mathbf W_{p_x}})}  \\
   {} & {\tr(\mathbf W_1^2)} & {\tr({\mathbf W_1}{\mathbf W_2})} & {...} & {\tr({\mathbf W_1}{\mathbf W_{p_x}})}  \\
   {} & {} & {\tr(\mathbf W_2^2)} & {...} & {\tr({\mathbf W_2}{\mathbf W_{p_x}})}  \\
   {} & {} & {} &  \ddots  &  \vdots   \\
   {} & {} & {} & {} & {\tr(\mathbf W_{p_x}^2)}
 \end{array} } \right) },
 \end{equation}
where ${\mathbf W_l} = {\dot {{\mathbf R}}_l}\mathbf Q$, for $1\leq  l \leq {p_x}$, and ${\dot {{\mathbf R}} _l}$ is  the {{partial}} derivative of  the correlation matrix $\mathbf R$  with respect to the $l$th range parameter, and $\mathbf { Q} =   \mathbf { R}^{ - 1} -  \mathbf { R}^{ - 1} {\mathbf H} ({{\mathbf H^T }{ \mathbf R^{ - 1}}{\mathbf H } })^{ - 1} \mathbf H^T {{\mathbf { R}}^{ - 1}}$, with $ \mathbf H=(\mathbf h^T(\mathbf x_1), ..., \mathbf h^T(\mathbf x_n))^T$.


After marginalizing out $(\bm \theta_m, \sigma^2)$ by the prior in (\ref{equ:refprior}), the marginal likelihood will be denoted as ${L}(\bm{\gamma } \mid \mathbf y )$.  As each evaluation of the likelihood requires the inversion of the covariance matrix, which is generally at the order of $O(n^3)$, the full Bayesian computation through the Markov Chain Monte Carlo (MCMC) is typically prohibitive. It is common to simply estimate $\bm \gamma$ by the marginal posterior mode in emulation
 \begin{equation}
 ({\hat \gamma}_1, \ldots {\hat \gamma}_{p_x})=  \mathop{\argmax }\limits_{ \gamma_1,\ldots, \gamma_{p_x}} \left\{ L( { \gamma_1}, \ldots, {\gamma_{p_x}} \mid \mathbf{y} )  \, \pi^R({\gamma_1}, \ldots, {\gamma_{p_x}})\right\}.
 \label{equ:est_gamma}
 \end{equation}

Some routinely used estimators, such as the maximum likelihood estimator (MLE) and maximum marginal likelihood estimator (MMLE) with regard to ${L}(\bm{\gamma } \mid \mathbf y )$ have been found to be unstable in estimating the range parameters  in various studies (see e.g. Figure 2 in \cite{li2005analysis}, Figure 2.2 in \cite{lopes2011development} and Figure 1 in \cite{Gu2018robustness}). The problem is often caused by the estimated correlation matrix being near-diagonal  ($\mathbf {\hat R} \approx \mathbf I_n$, where $\mathbf I_n$ is the identity matrix of size $n$) or being near-singular ($\mathbf {\hat R} \approx \mathbf 1_n \mathbf 1^T_n$). In both scenarios, the problems are caused by the estimation of the range parameters. More specifically, as shown in Lemma 3.3. in \cite{Gu2018robustness}, the profile likelihood function may not decrease when any $\gamma_l \to 0$, $l=1,...,p_x$, which sometimes results in the estimated correlation matrix being near diagonal, whereas the marginal likelihood may not decrease when any $\gamma_l \to 0$ or all $\gamma_l \to 0$, $l=1,...,p_x$, leading to the near-diagonal correlation matrix or  near-singular correlation matrix, respectively. Thus, the robust estimation of the  parameters is defined as avoiding these two possible problems, as follows.
  \begin{definition} (Robust Estimation.)
 Estimation of the parameters in the GaSP is called robust, if neither $ \hat {\mathbf R}= \mathbf 1_n \mathbf 1^T_n$ nor  $\hat {\mathbf R} =\mathbf I_n $, where $\hat {\mathbf R}$ is the estimated correlation matrix.
\label{def:robustness}
\end{definition}

It is shown in \cite{Gu2018robustness} that the marginal posterior mode estimation with the reference prior is robust under $ \gamma$ or $ \xi=\log(1/ \gamma)$ parameterization, while some other alternatives, such as the MLE and MMLE, do not have this property.  Note that the near-diagonal estimation ($ \hat {\mathbf R}\approx \mathbf I_n$)  can easily happen for ${p_x}>1$ when a product correlation structure is used because, if any of the matrices in the product correlation matrix is near-diagonal, the correlation matrix will be near-diagonal. Thus, using the maximum marginal posterior mode estimation with the  reference prior is particularly helpful, when the dimension of the input is larger than 1. 





\begin{figure}[t]
\centering
  \begin{tabular}{ccc}
    \includegraphics[height=.3\textwidth,width=.325\textwidth]{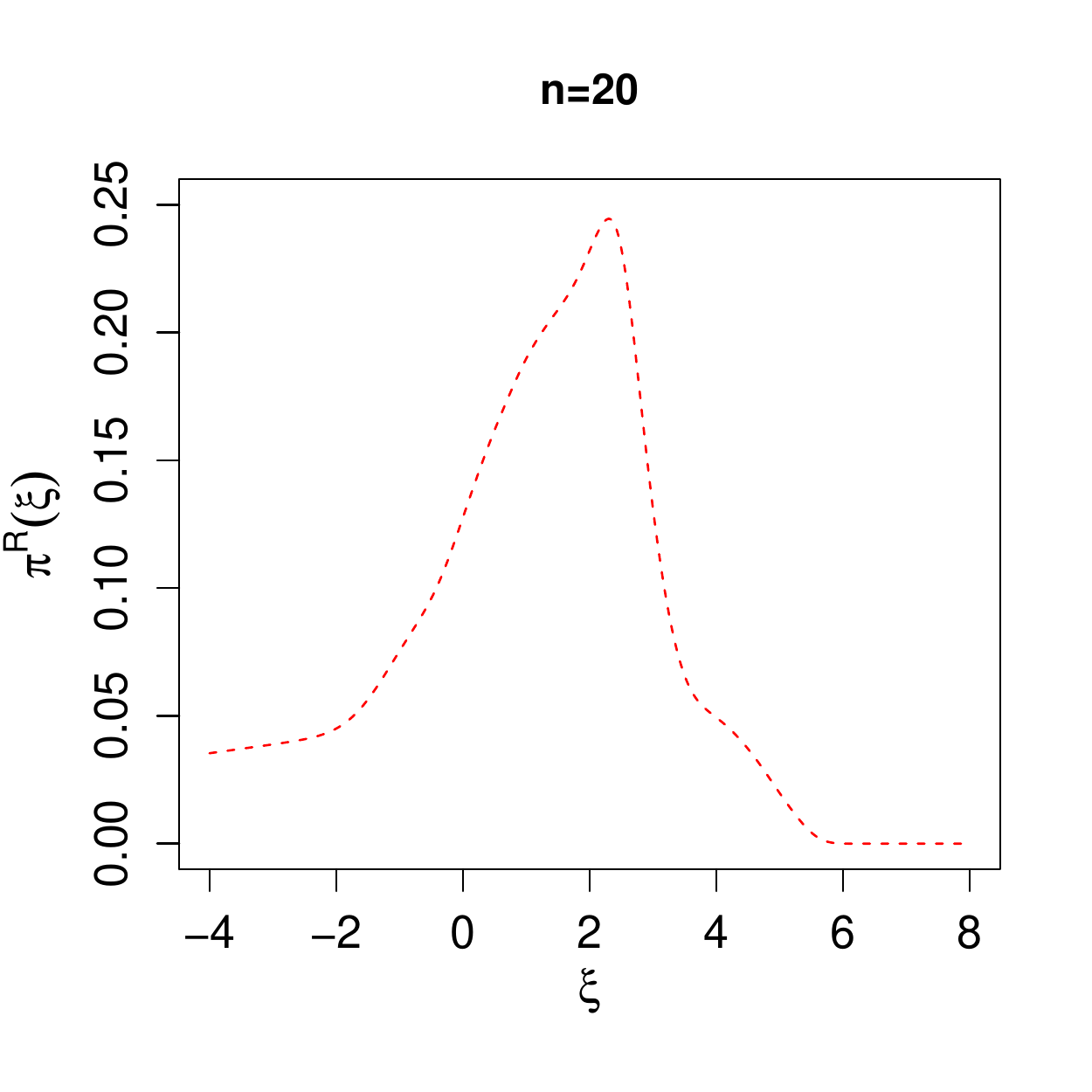}
        \includegraphics[height=.3\textwidth,width=.325\textwidth]{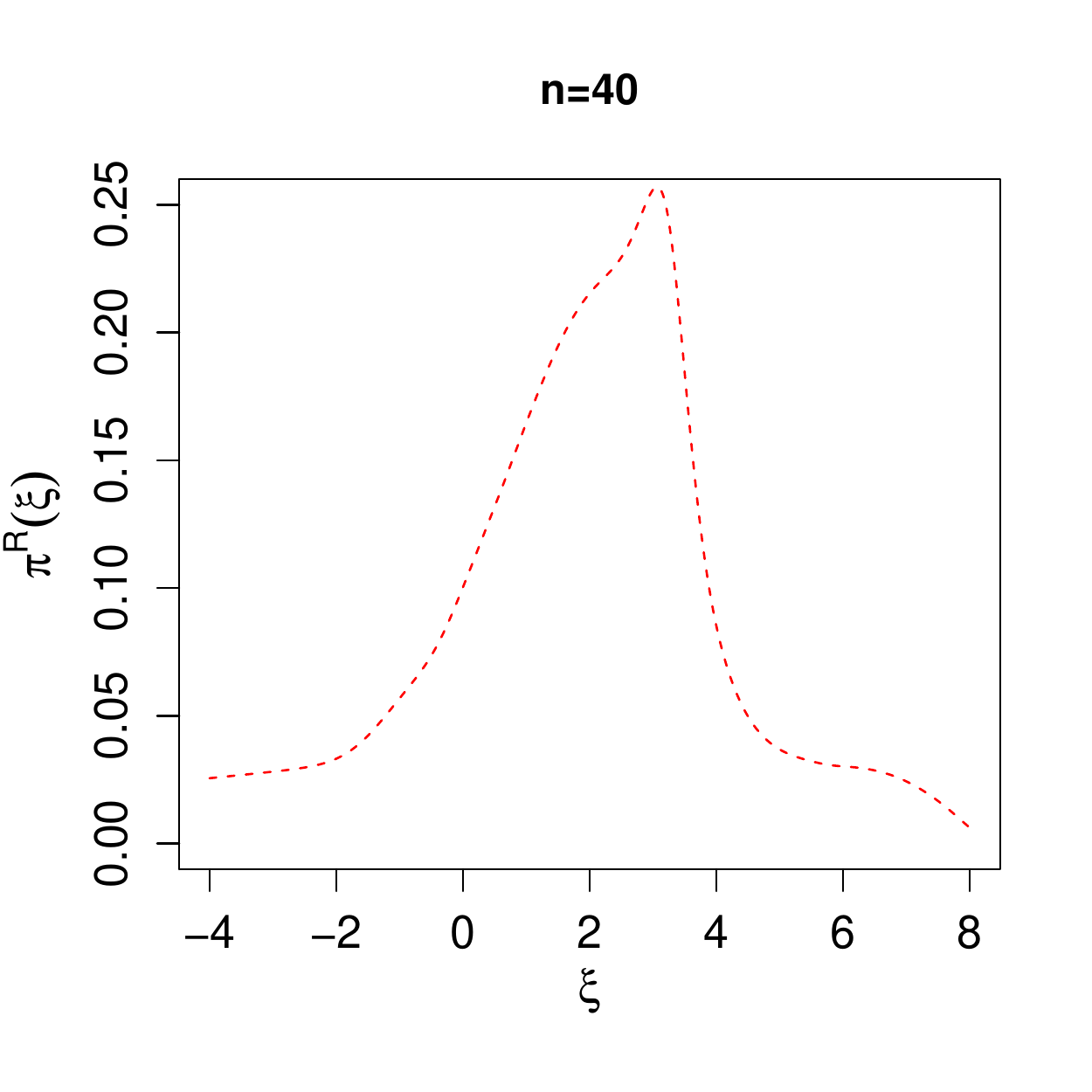}
                \includegraphics[height=.3\textwidth,width=.325\textwidth]{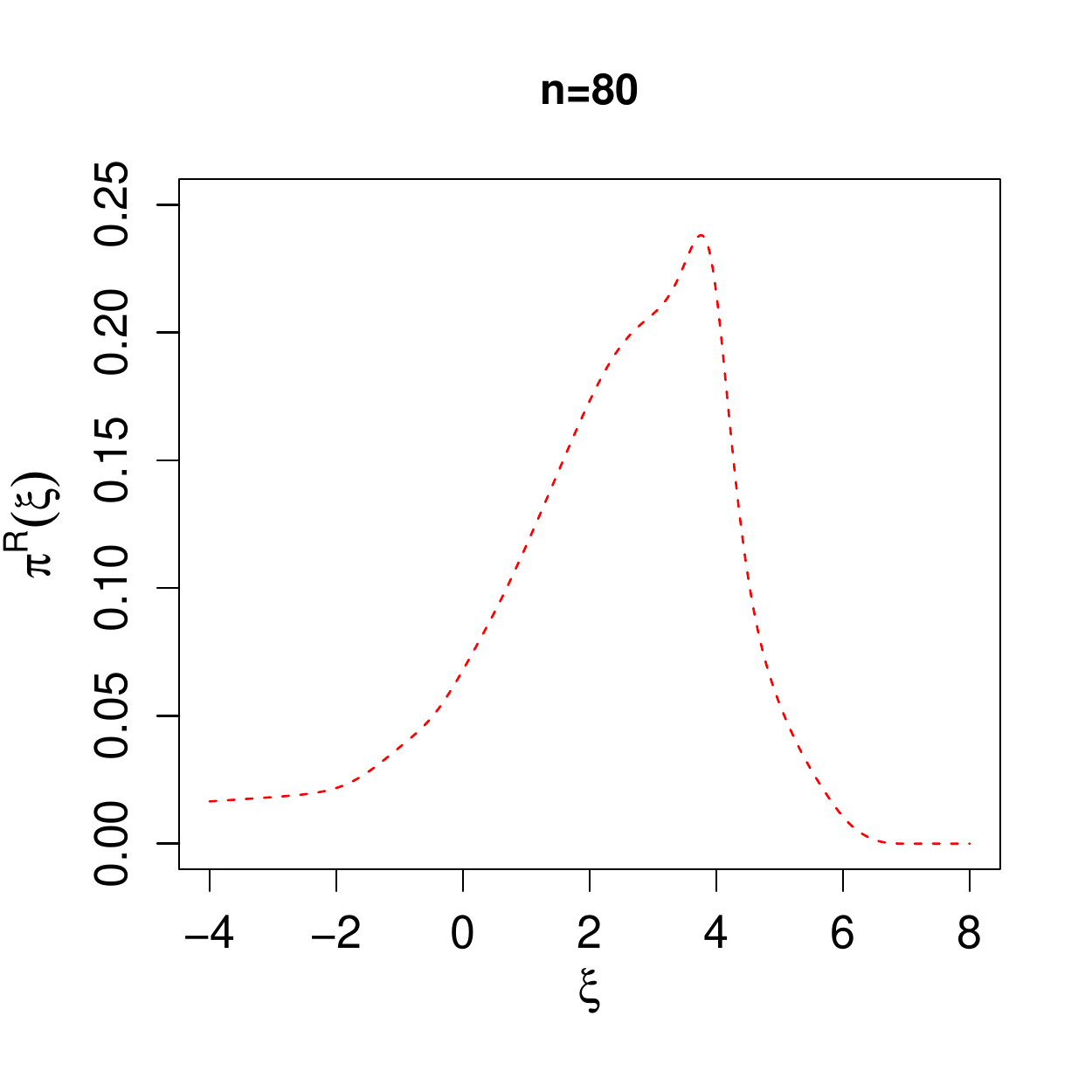} \vspace{-.15in}\\
    \includegraphics[height=.3\textwidth,width=.314\textwidth]{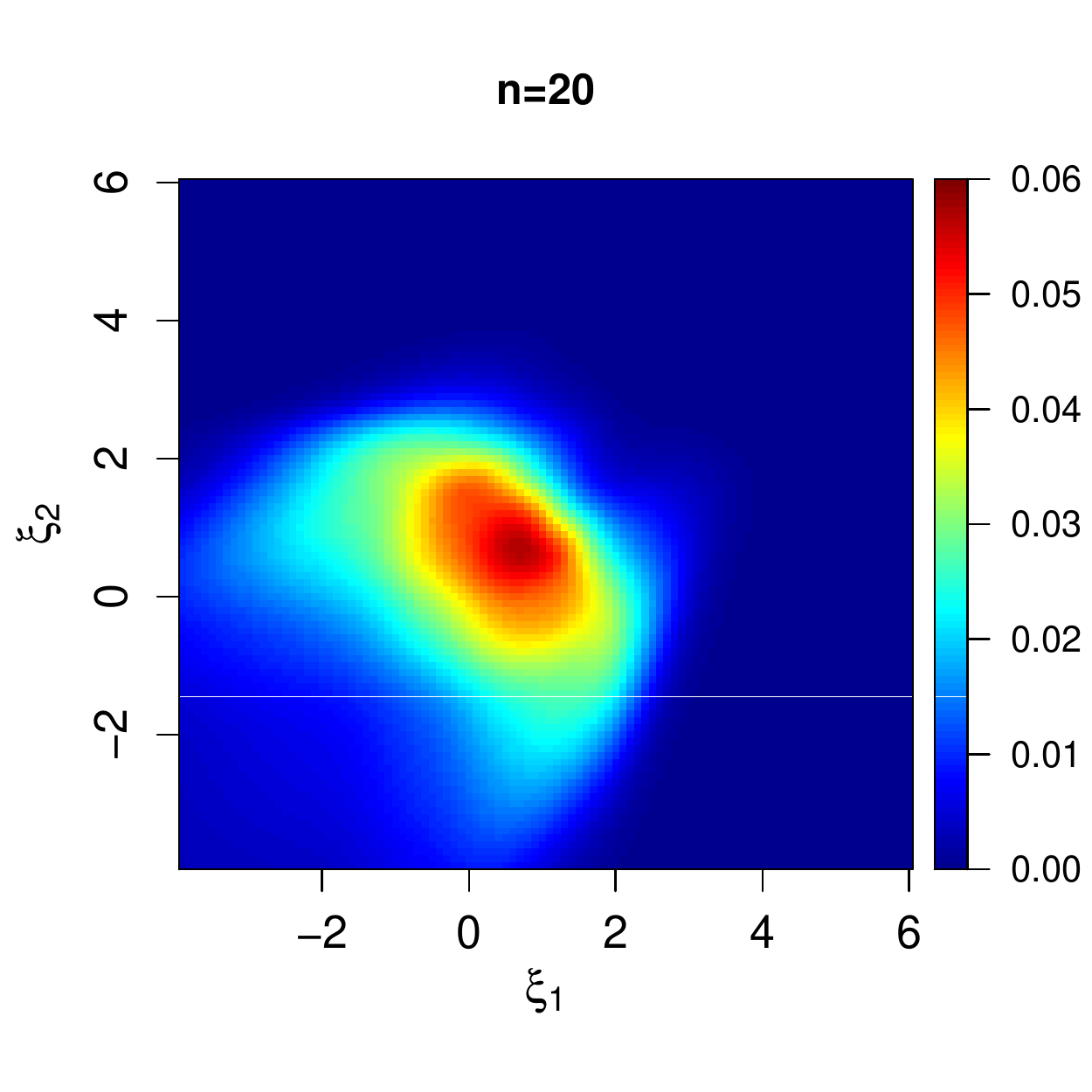}
        \includegraphics[height=.3\textwidth,width=.314\textwidth]{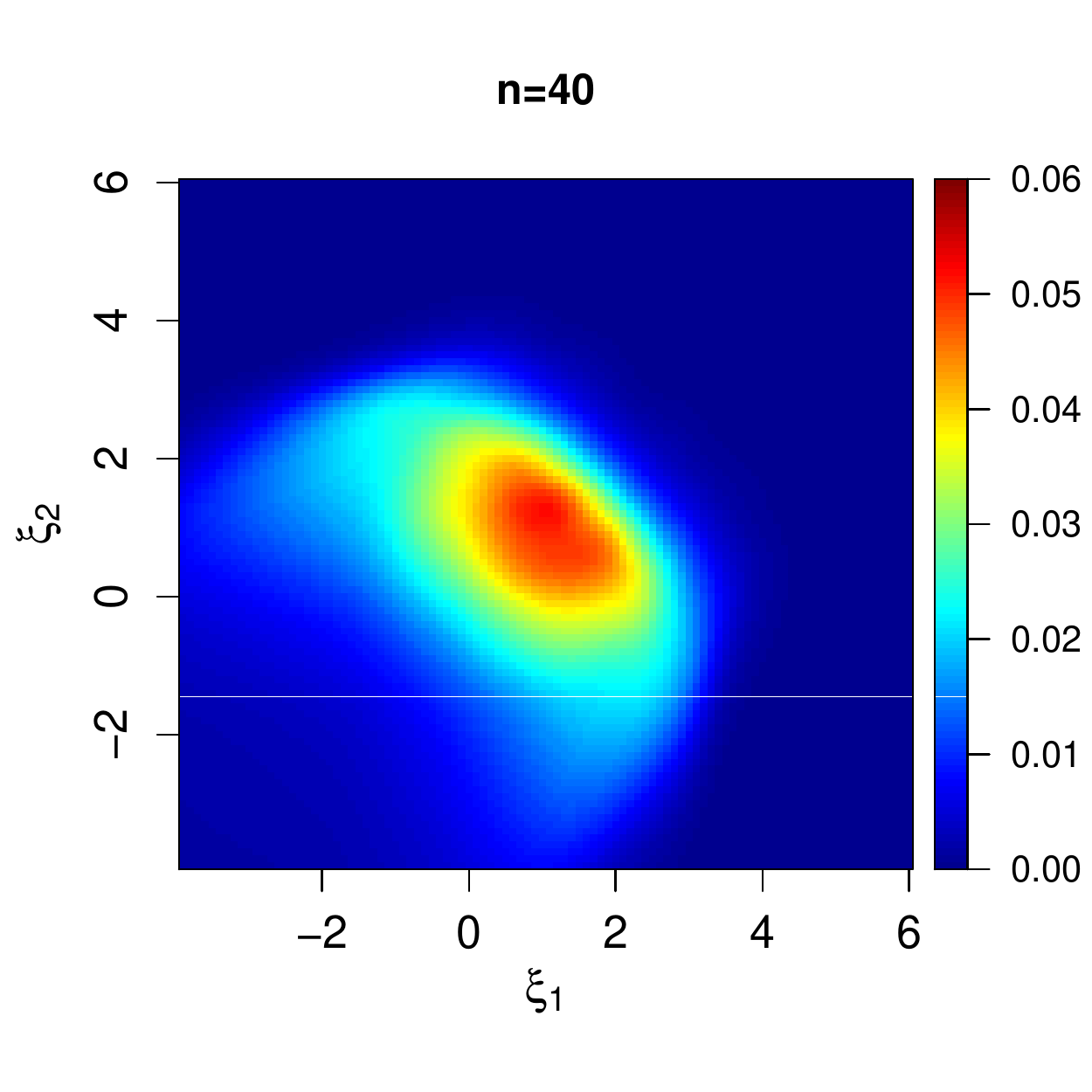}
                \includegraphics[height=.3\textwidth,width=.314\textwidth]{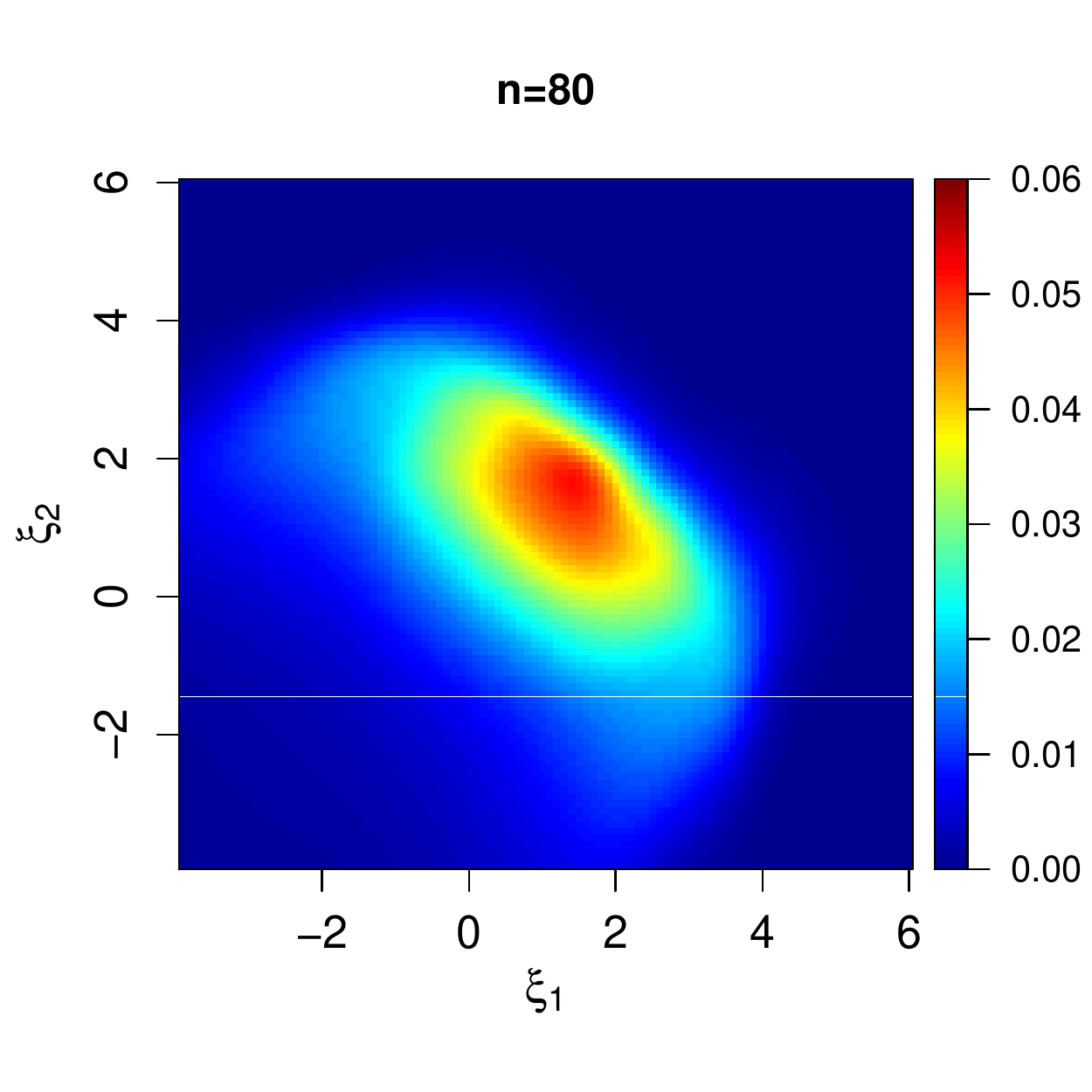}
  \end{tabular}
   \caption{Density of the reference prior of the log inverse range parameter (up to the normalizing constants). The power exponential correlation function in (\ref{equ:pow_exp}) is assumed where $\alpha_l=1.9$, $1\leq l\leq {p_x}$, with ${p_x}=1$ (upper panels) and ${p_x}=2$ (lower panels). From left to right, the number of design points are $n=20$, $n=40$ and $n=80$, respectively, all generated from a maximin Latin Hypercube (LHD) on $[0,1]^{p_x}$ (\cite{santner2003design}). For all the panels, we assume $\mathbf H=\mathbf 1_n$.   \vspace{-.1in}}
\label{fig:reference_prior_1dim_2dim}
\end{figure}


 The reference prior has many other advantages in emulation that were not noticed before. First, when the dimension of the inputs  increases, the prior mass moves from the smaller values of $\gamma_l$ to the large values of $\gamma_l$, for each $l=1,...,{p_x}$. This is an important property since, as any of $\hat \gamma_l \approx 0$, $\hat {\mathbf R}$ is near diagonal, a degenerate case that should be avoided. When ${p_x}$ increases, the chance that at least one  $\gamma_l$ is estimated to be small increases, if the prior mass does not change along with ${p_x}$ and, consequently, the chance that $\hat {\mathbf R}\approx \mathbf I_n$  also increases. The reference prior adapts to the increase of the dimension by concentrating more prior mass at larger $\gamma_l$, avoiding $\hat {\mathbf R}\approx \mathbf I_n$, when ${p_x}$ increases. 
 

Second, when a denser design is used in a fixed domain of the input space, the prior mass of the reference prior parameterized by $\gamma_l$ moves to the domain with smaller values. This is helpful for the inversion of the covariance matrix in practice, because as points fill with a fixed domain of the input space, the covariance matrix becomes singular if $\hat \gamma_l$ does not change. 


%
%
%
Here we provide a numerical justification of first two properties of the reference prior in Figure~\ref{fig:reference_prior_1dim_2dim}, where the reference prior density of the log inverse range parameter $\xi_l=\log(1/\gamma_l)$  for the GaSP with a power exponential correlation function is shown.  Comparing the figures with different sample sizes,   the mode of the prior moves to the region with larger values of log inverse range parameters (or equivalently the smaller range parameters), when the sample size increases. Comparing the figures with different dimensions, the prior mass moves to the region with smaller values of the log inverse range parameters (or equivalently the larger range parameters), when the inputs have higher dimensions. 




The third property of interest is that the reference prior is invariant to the location-scale transformation of the inputs, if the mean basis functions contain only the intercept and the linear terms of $\mathbf x$.  When we apply a location-scale transformation of each coordinate of the input $\tilde x_{l}=\frac{x_{l}- {c_{0l}}}{c_{1l}}$, for $l=1,...,{p_x}$, the new  reference prior  is $\tilde \pi^R(\gamma_1,...,\gamma_{p_x})=\pi^R(\gamma_1/c_{11},...,\gamma_p/c_{1{p_x}})$. This makes the prior scale naturally to the range of the inputs; as a consequence, we do not need to normalize the inputs.


 In addition, the reference prior has an appropriate tail decay rate at the limits when $\mathbf R = \mathbf I_n$ and $\mathbf R = \mathbf 1_n\mathbf 1^T_n$ (\cite{Gu2018robustness}).  When  $ \gamma_l \to 0$ for any $l=1,...,{p_x}$, the density of the reference prior decreases at an exponential rate approximately; when $ \gamma_l \to \infty$ for all $l=1,...,{p_x}$, the density of the  reference prior deceases at a polynomial rate. The first part of the tail rates induces an exponential penalty to the likelihood when the correlation matrix is near diagonal, prohibiting the undesired situation in emulation. The posterior with the reference prior has slow polynomial decay rates when $ \gamma_l$ is large for all $l=1,...,{p_x}$ (or equivalently $\mathbf R \approx \mathbf 1_n\mathbf 1^T_n$), allowing the marginal likelihood to come into play at this limit. The larger $\gamma_l$ is found to make prediction more precise (\cite{zhang2004inconsistent}), and thus a small polynomial penalty from the reference prior both reduces the singular estimation of the covariance matrix and maintains high accuracy in prediction.

 Despite various benefits in using the reference prior for emulation,  the computational challenges still persist with the use of the reference prior, even if the posterior mode estimation is used in lieu of the posterior sampling. The computational order of the reference prior is $O(p_xn^3)$, which is mainly from computing $\mathbf W_l$ in (\ref{equ:efi}), for $l=1,...,p_x$, and the inversion of the covariance matrix.  However,  the closed form derivatives of the reference prior are very computationally intensive, as it requires to compute $\partial^2 \mathbf R/ \partial \gamma_i \partial \gamma_j $, for $1\leq i,\,j\leq p_x$. The total computational orders of $p_x$ directional derivatives of the reference prior is $O(p^3_xn^3)$, because the computational order of each directional derivative is $O(p^2_xn^3)$ for the matrix multiplication. Because of these reasons, the author does not find any literature that provides the closed-form derivatives of the reference prior in this scenario, though some frequently used mode search algorithms, such as the low-storage quasi-Newton optimization method (\cite{nocedal1980updating}), typically rely on the information of the derivatives. Instead, one typically computes the numerical derivatives, which requires more evaluations of the likelihood, each with $O(n^3)$ in computing the inversion of the covariance matrix, and thus it is also very time consuming. In addition, the reference prior could also induce some extra local modes, making the optimization algorithm harder to converge to the global mode (see e.g. the change of the slope of the reference prior density in the upper middle panel in Figure \ref{fig:reference_prior_1dim_2dim} and another example is given in Figure 3.3 in \cite{Gu2016thesis}).



Some inputs of the computer model may have very small effects on the outputs of the computer model. These inputs are called inert inputs and are often omitted in emulation. When the inert inputs are omitted, a noise is needed in the GaSP emulator, as the emulator should no longer be an interpolator at the design points. The GaSP emulator can be extended to include a  noise or nugget, $\tilde f^M(\cdot)=f^M(\cdot)+\epsilon$, where $f^M(\cdot)$ still follows a GaSP model and $\epsilon \sim N(0, \sigma^2_0)$ is an independent Gaussian noise.   Define the nugget variance ratio parameter $\eta=\sigma^2_0/\sigma^2$. The reference prior $\pi^R(\bm \gamma, \eta)$ has been derived for the GaSP model with a noise  (\cite{ren2012objective,kazianka2012objective,Gu2016PPGaSP}). The advantages of using the reference prior with a nugget are similar to our previous discussion and are thus omitted here.


\subsection{GaSP for computer model calibration}
\label{subsec:GP_calibration}
Some parameters in the computer model are unknown and unobservable in experiments. We denote the mathematical model output by $f^M(\mathbf x, \bm \theta)$, where $\mathbf x$ is a ${p_x}$-dimensional vector of  observable inputs in experiment and $\bm \theta$ is a ${p_{\theta}}$-dimensional vector of  unobservable parameters. The calibration problem is to estimate $\bm \theta$ by a set of field data $\mathbf y^F:=(y^F(\mathbf x_1),...,y^F(\mathbf x_n))^T$. In practice, a perfect mathematical model to the reality is rarely the case. It is common to address the model misspecification by a discrepancy function, such that the reality can be represented as  $y^R(\mathbf x)=f^M(\mathbf x, \bm \theta)+\delta(\mathbf x)$, where $y^R(\cdot)$ and $\delta(\cdot)$ denote the reality and discrepancy function, respectively. It leads to the following statistical model for calibration
 \begin{equation}
y^F(\mathbf x)=f^{M}(\mathbf x, \bm \theta)+\delta(\mathbf x)+\epsilon,
\label{equ:model_calibration}
\end{equation}
where $\epsilon \sim N(0, \, \sigma^2_0)$ is an independent zero-mean Gaussian noise.  For simplicity, we assume $f^M(\cdot, \cdot)$ is computationally cheap to evaluate for now.

 As we often know very little  about the discrepancy function,  the GaSP is suggested in \cite{kennedy2001bayesian} to model the discrepancy function, i.e. $\delta(\cdot) \sim \mbox{GaSP}(\mu(\cdot), \, \sigma^2 c(\cdot, \cdot) )$, where the mean and correlation functions are defined in (\ref{equ:mean}) and (\ref{equ:correlation}), respectively. It is usual to define $\eta=\sigma^2_0/\sigma^2$, the nugget-variance ratio parameter for the computational reason, as now $\sigma^2$ is a scale parameter which has a conjugate prior.

The parameters in (\ref{equ:model_calibration}) consist of the calibration parameters, mean parameters,   range parameters, a variance parameter and a nugget parameter in the covariance function, denoted as $(\bm \theta, \bm \theta_m, \bm \gamma, \sigma^2, \eta)$.    Consider the following prior for the calibration problem
\begin{equation}
\pi(\bm \theta, \bm \theta_m, \sigma^2, \bm \gamma, \eta ) \propto \frac{\pi(\bm \gamma, \eta) \pi(\bm \theta) }{\sigma^2}.
\label{equ:prior_calibration}
\end{equation}
As the calibration parameters normally have scientific meanings, $\pi(\bm \theta)$ is typically chosen by expert knowledge and thus we do not give a specific form herein. To the author's knowledge, the posterior propriety has not been shown for the above prior in the calibration problem, except for the case that $f^M(\mathbf x, \bm \theta)$ is linear with regard to $\bm \theta$. We have the following theorem to guarantee the posterior propriety when the prior in (\ref{equ:prior_calibration}) is used. 	The proof	for Theorem \ref{thm:post_propriety} generalizes the proof  in \cite{berger1998bayes}, which is a special case with a mean parameter, a variance parameter and two independent observations.  



\begin{theorem}
Assume the prior  follows (\ref{equ:prior_calibration}) for the calibration model in (\ref{equ:model_calibration}) with $\pi(\bm \gamma, \eta)$ and $\pi(\bm \theta)$ being proper priors.  Let $\mathbf H_y=(\mathbf H, \mathbf y^F)$ be an $n\times (q+1)$ matrix. If $\mathbf H_y$ has full rank and $n\geq q+1$,  the posterior  is proper. 

\label{thm:post_propriety}
\end{theorem}
\begin{proof}
Since $\mathbf H_y$ has full rank and $n\geq q+1$, one can select $q+1$ linearly independent rows of $\mathbf H_y$, denoted as  $\mathbf H_{y0}$, such that $\mathbf H_{y0}$ is invertible. W.l.o.g., we assume the first $q+1$ rows of $\mathbf H_y$ are linearly independent. 

 We first marginalize out the last $n-q-1$ field data and the resulting density is denoted as $p(\mathbf y^F_{1:(q+1)}\mid \bm \theta_m, \sigma^2, \bm \theta, \bm \gamma, \eta)$.  As $\pi(\bm \theta)$ and $\pi(\bm \gamma, \eta)$ are both proper, we then marginalize out $(\bm \theta, \bm \gamma, \eta)$ and obtain the proper marginal density $p(\mathbf y^F_{1:(q+1)}\mid \bm \theta_m, \sigma^2)$. Since $(\bm \theta_m, \sigma^2)$ are the location-scale parameters for the marginal density, one has 
\begin{align*}
&\int ...\int p(\mathbf y^F_{1:(q+1)}\mid \bm \theta_m, \sigma^2) \pi(\bm \theta_m, \sigma^2) d \bm \theta_m d\sigma^2 \\
=&\int ...\int \frac{1}{ (\sigma^2)^{(q+1)/2+1} } p\left(\frac{y(\mathbf x_1) - \mathbf h(\mathbf x_1)\bm \theta_m}{\sigma},...,\frac{y(\mathbf x_{q+1}) - \mathbf h(\mathbf x_{q+1}) \bm \theta_m}{{\sigma}} \right)d \bm \theta_m d\sigma^2  \\
=&\int ... \int \frac{1}{ (\sigma^2)^{(q+1)/2+1} } | J^{-1} |  p\left( \tilde y_1,..., \tilde y_{q+1} \right)d \tilde y_1...d \tilde y_{q+1}
\end{align*}
where the fist equation follows from the definition of the location-scale family and the second equation follows from parameter transformation for $\tilde y_i=\frac{y(\mathbf x_i)- \mathbf h(\mathbf x_i)\bm \theta_m }{\sigma}$, for $i=1,...,{q+1}$, with the Jacobian determinant being
		\begin{align*}
                    J^{-1}
		&=\left|\begin{matrix}
		-\frac{h_1(\mathbf x_1)}{\sigma} & \cdots &
		-\frac{h_{q}(\mathbf x_1)}{\sigma}&-\frac{y^F(\mathbf
			x_1) -\mathbf h(\mathbf x_1) \bm \theta_m }{\sigma^3}\\
		\vdots&\ddots&\vdots&\vdots\\
		-\frac{h_1(\mathbf x_{q+1})}{\sigma} & \cdots &
		-\frac{h_{q}(\mathbf
			x_{q+1})}{\sigma}&-\frac{y^F(\mathbf x_{q+1})-\mathbf h(\mathbf
			x_{q+1}) \bm \theta_m }{\sigma^3}
		\end{matrix}\right|^{-1} \nonumber\\
		&={\sigma^{q+3}}\left|\begin{matrix}
		h_1(\mathbf x_1) & \cdots & h_{q}(\mathbf
		x_1)&y^F(\mathbf x_1)\\
		\vdots&\ddots&\vdots&\vdots\\
		h_1(\mathbf x_{{q+1}}) & \cdots & h_{q}(\mathbf
		x_{{q+1}})&y^F(\mathbf x_{{q+1}})		\end{matrix}
		\right|^{-1} ,\\
		&= {\sigma^{q+3}}  J_0^{-1},
		\label{equ:jacobian}
		\end{align*}
    where $J_0=|\mathbf H_{y0}|$. 
		Hence one has 
		\[ \int ...\int p(\mathbf y^F_{1:(q+1)}\mid \bm \theta_m, \sigma^2) \pi(\bm \theta_m, \sigma^2) d \bm \theta_m d\sigma^2 d \mathbf y^F_{1:(q+1)}= | J_0^{-1}| <\infty.  \]
		

\end{proof}

Note that the reference prior in (\ref{equ:refprior}) is proper for many widely used correlation functions, as long as the intercept is contained in the mean basis matrix, i.e. $\mathbf 1_n \in \mathcal C(\mathbf H)$, where $\mathcal C(\mathbf H)$ denotes  the column space of the mean basis matrix  (\cite{Gu2018robustness}).  Theorem~\ref{thm:post_propriety} states that using the reference prior is legitimate in the calibration problem when the mean basis contains an intercept.  Empirically, the reference prior changes very little with an added intercept in the column space of the mean basis matrix.

We have assumed the mathematical model is computationally cheap so far. When the computer model is expensive to run, one can combine the GaSP emulator in the calibration model through a full Bayesian approach.  In practice, however, since the field data typically contain larger noises and may not provide much information for the emulation purpose, a modular approach is often used,  meaning that the GaSP emulator only depends on the outputs of the computer model (\cite{liu2009modularization}). We refer to  \cite{bayarri2007framework} for an overview of combining the GaSP emulator in a calibration model. The modular approach is implemented in \cite{gu2018robustcalibrationpackage}, where the parameters in the GaSP emulator are estimated based on the outputs of the computer model (\cite{gu2018robustgasp}). In the calibration, we draw a sample from the posterior predictive distribution of the GaSP emulator when we need to evaluate the computer model.

\section{Jointly robust prior}
\label{sec:JR_prior}
We introduce a new class of priors for  calibration and emulation of  mathematical models in this section. In Section \ref{subsec:JR_Calibration} and Section \ref{subsec:emulation_JR_prior}, we show that the JR prior has all the nice features of the reference prior discussed in Section \ref{subsec:GP_emulator}. The benefits of the new prior in calibration and identifying the inert inputs in mathematical models will be discussed in  Section  \ref{subsec:JR_Calibration} and Section \ref{sec:variable_selection}, respectively.

\subsection{Calibration}
\label{subsec:JR_Calibration}

 We first introduce the new prior in the calibration setting, where the model in given in (\ref{equ:model_calibration}) with the discrepancy function $\delta(\cdot)$ modeled as a GaSP.  Define the inverse range parameter $\beta_l=1/\gamma_l$, for $l=1,...,{p_x}$, and the  nugget-variance parameter $\eta=\sigma^2_0/\sigma^2$ in the covariance function. The overall prior follows
\begin{equation}
\pi(\bm \theta, \bm \theta_m, \sigma^2, \bm \beta, \eta ) \propto \frac{\pi^{JR}(\bm \beta, \eta) \pi(\bm \theta) }{\sigma^2}.
\label{equ:JR_all}
\end{equation}
The key part is the prior for the range parameters and nugget-variance parameter, where we call it the {jointly robust (JR) prior} and the form is given as follows
\begin{equation}
\pi^{JR}(\beta_1,...,\beta_{p_x}, \eta)=C\left(\sum^{{p_x}}_{l=1} C_l \beta_l+ \eta \right)^{a} \exp\left\{-b\left(\sum^{{p_x}}_{l=1} C_l\beta_l + \eta\right)\right\}_,
\label{equ:JRprior}
\end{equation}
where $C$ is a normalizing constant; $a>-({p_x}+1)$, $b>0$ and $C_l>0$ are prior parameters. The name ``jointly robust"  is used to reflect the fact that the prior can't be written as a product of the marginal priors of the range parameter for each coordinate of the input, and it is robust in marginal posterior mode estimation (see Section~\ref{subsec:emulation_JR_prior} for details).  The form $\sum^{{p_x}}_{l=1} C_l \beta_l+ \eta$ is inspired by the tail rate of the reference prior at $\mathbf R = \mathbf 1_n\mathbf 1^T_n$ shown in Lemma 4.1 in \cite{Gu2018robustness}.  Besides, the JR prior is a proper prior. The posterior propriety of using (\ref{equ:JRprior}) is thus guaranteed when $\pi(\bm \theta)$ is proper, shown in Theorem~\ref{thm:post_propriety}.

 We first show some properties of this prior and then discuss the default choice of the prior parameters. First of all, the normalizing constant of the prior is given as follows.
\begin{lemma} (Normalizing constant.)
 The jointly robust prior is proper and has the normalizing constant $C=\frac{{p_x}! b^{a+{p_x}+1}\prod^{{p_x}}_{l=1}C_l }{\Gamma(a+{p_x}+1) }$, where $\Gamma(\cdot)$ is the gamma function.
 \label{lemma:norm_const}
\end{lemma}

\begin{proof}[Proof of Lemma~\ref{lemma:norm_const}]
\begin{align*}
\frac{1}{C}=&\int...\int(\sum^{{p_x}}_{l=1}C_l\beta_l+ \eta)^aexp(-b(\sum^{{p_x}}_{l=1}C_l\beta_l +\eta) ) d\eta d\beta_1...d\beta_{p_x} \\
 =& \int...\int\frac{(\sum^{{p_x}}_{l=1}\tilde\beta_l +\eta )^aexp(-b(\sum^{{p_x}}_{l=1}\tilde\beta_l +\eta))}{\prod^{p_x}_{l=1}C_l}d \eta {d\tilde \beta_1...d\tilde \beta_{p_x}},  \quad \quad {\rm let} \, \tilde \beta_l=C_l\beta_l, \\
=&\int \frac{z^a \exp(-bz)}{\prod^{p_x}_{l=1}C_l}\int...\int_{\tilde \beta_1+...+\tilde\beta_{p_x}<z} {d\tilde \beta_1...d\tilde \beta_{p_x}} dz,  \quad \quad {\rm let} \, z=\sum^{{p_x}}_{l=1}\tilde \beta_l+\eta,  \\
=&\int \frac{z^a \exp(-bz)z^{p_x}}{\prod^{p_x}_{l=1}C_l {p_x}!} dz, \\
=&\frac{\Gamma(a+{p_x}+1)}{ {p_x}!b^{a+{p_x}+1} \prod^{p_x}_{l=1}C_l}.
\end{align*}
\end{proof}
The marginal prior mean and variance are given in the following lemma. 
\begin{lemma} (Prior mean and variance.)
For $i=1,...,{p_x}$, the prior mean and prior variance are given as follows.
\item[(i.)] ${\E}_{\pi^{JR}}[\beta_i]=\frac{a+{p_x}+1}{ ({p_x}+1) C_ib}$ and ${\E}_{\pi^{JR}}[\eta]=\frac{a+{p_x}+1}{ ({p_x}+1) b}$. 
\item[(ii.)] ${\Var}_{\pi^{JR}}[\beta_i]=\frac{(a+{p_x}+1)\{ ({p_x}+1)^2+{p_x}+a{p_x}+1\}}{({p_x}+1)^2({p_x}+2)C^2_ib^2}$ and  ${\Var}_{\pi^{JR}}[\eta]=\frac{(a+{p_x}+1)\{ ({p_x}+1)^2+{p_x}+a{p_x}+1\}}{({p_x}+1)^2({p_x}+2) b^2}$.
\label{lemma:prior_mean}
\end{lemma}
\begin{proof}[Proof of Lemma~\ref{lemma:prior_mean}]  We only show the prior mean for $\beta_i$, as the proof of the prior mean for $\eta$ is similar. For any $1\leq i\leq {p_x}$
\begin{align*}
\E[\beta_i]=&\int...\int \beta_i (\sum^{{p_x}}_{l=1}C_l\beta_l+ \eta)^a \exp\left(-b(\sum^{{p_x}}_{l=1}C_l\beta_l +\eta) \right) d\eta d\beta_1...d\beta_{p_x} \\
 =& \int \frac{c z^a \exp(-bz)}{C_i\prod^{p_x}_{l=1}C_l}  \int...\int_{\tilde \beta_1+...+\tilde\beta_{p_x}<z} \tilde \beta_i{d\tilde \beta_1...d\tilde \beta_{p_x}} dz,  \quad  {\rm let} \, \tilde \beta_l=C_l\beta_l, \, z= \sum^{p_x}_{l=1}\tilde \beta_l+\eta\\
 =& \int \frac{c z^a \exp(-bz)}{C_i\prod^{p_x}_{l=1}C_l}  \int^{z}_{0}\frac{(z-\tilde \beta_i)^{{p_x}-1}\tilde \beta_i}{({p_x}-1)!}d\tilde \beta_i dz \\
 =& \frac{a+{p_x}+1}{ ({p_x}+1) C_ib}.
\end{align*}
Using the similar method for the prior mean, for any $1\leq i\leq {p_x}$, we have 
\[\E_{\pi^{JR}}[\beta^2_i]=\frac{2(a+{p_x}+2)(a+{p_x}+1)}{({p_x}+1)({p_x}+2)C^2_i b^2}.\]
  Part (ii) follows from $\Var_{\pi^{JR}}[\beta_i]=\E_{\pi^{JR}}[\beta^2_i]-(\E_{\pi^{JR}}[\beta_i])^2 $ for $i=1,...,p_x$.
\end{proof}





The prior parameters of the jointly robust prior in Equation (\ref{equ:JRprior}) consist of the overall scale parameter $a$, the rate parameter $b$ and input scale parameters $C_l$, $l=1,...,{p_x}$. First, we let $C_l= n^{-1/{p_x}}|x^{max}_{l}-x^{min}_{l}| $, where $x^{max}_{l}$ and $x^{min}_{l}$ are the maximum and minimum values of the input at the $l$th coordinate, which makes the reference prior invariant to the location-scale transformation of the input. The factor $ n^{-1/{p_x}}$ is the average distance between the inputs, as the average sample size for each coordinate of the input is $n^{1/{p_x}}$ when we have $n$ inputs from a Lattice design at a $p$ dimensional input space. This choice allows the JR prior to match the behavior of the reference prior  to the change of dimensions and number of observations. Second, we let $b=1$ to have a large exponential penalty to avoid the estimation of $\mathbf R$ being near diagonal. 




The choice of $a$ is an open problem and may depend on specific scientific goals. In the calibration setting, when $a$ is close to $-1-{p_x}$, the prior density is almost flat when $\log(\beta)\to 0$ and $\log(\eta) \to 0$, resulting in the large estimated correlation in some scenarios, which makes the calibrated computer model without the discrepancy function fit  the reality poorly (\cite{gu2017improved}). On the contrary, when $a$ is large, the method is biased to small correlation and make the prediction less accurate. In the RobustCalibration package (\cite{gu2018robustcalibrationpackage}), $a=1/2-{p_x}$ is the default setting for the calibration problem, which balances between prediction and calibration. $a=1/2-{p_x}$, $b=1$ and  $C_l= n^{-1/{p_x}}|x^{max}_{l}-x^{min}_{l}| $, $l=1,...,p_x$, will be used for all numerical comparisons in calibration.




\begin{figure}[t]
\centering
  \begin{tabular}{ccc}
                \includegraphics[height=.3\textwidth,width=1\textwidth]{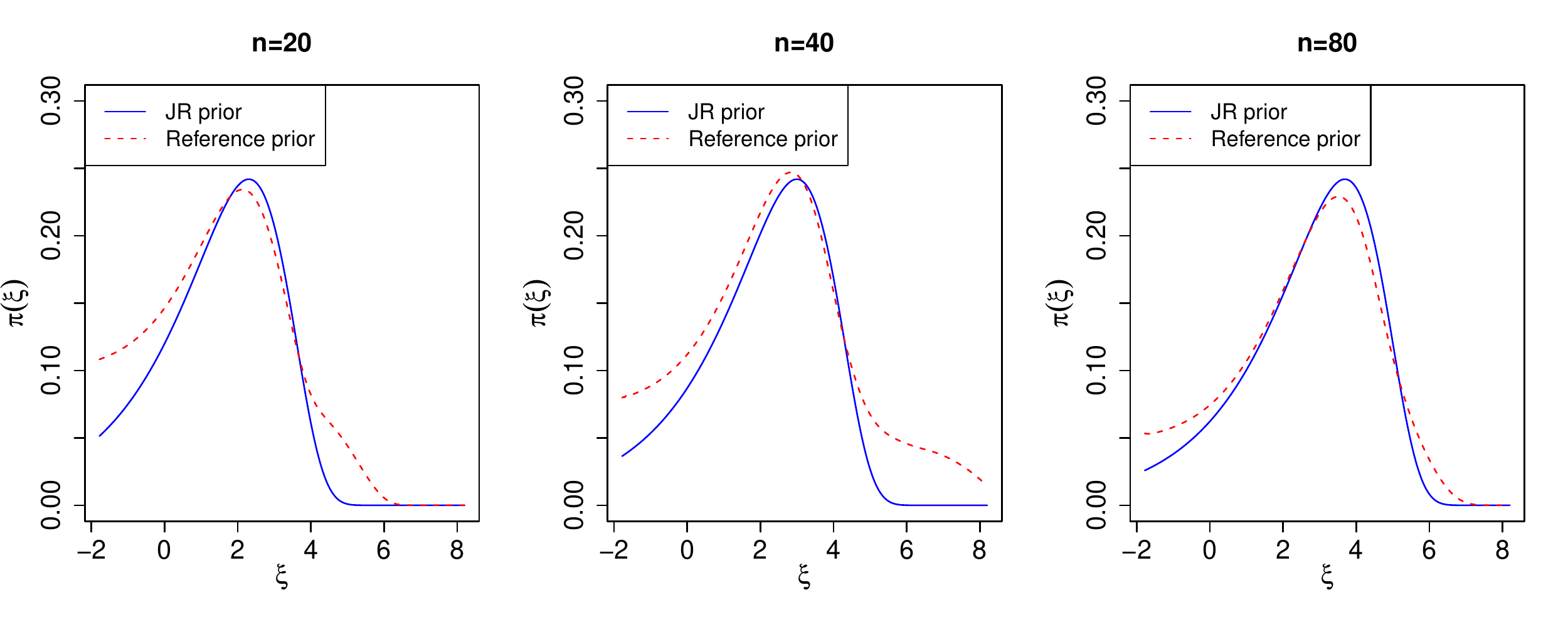} \vspace{-.15in}\\
                \includegraphics[height=.3\textwidth,width=.96\textwidth]{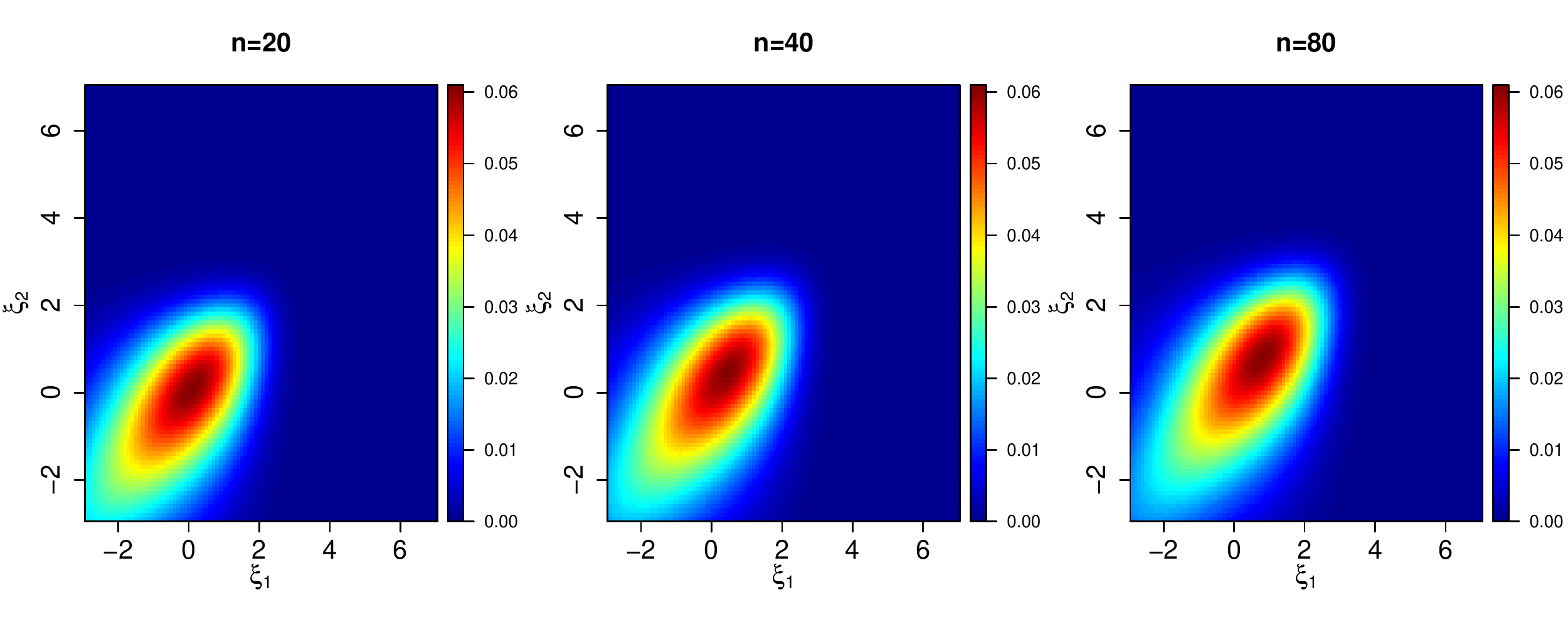} \vspace{-.15in}\\
  \end{tabular}
   \caption{Priors for the log inverse range parameter of the Mat{\'e}rn correlation function where $\alpha_l=2.5$, $1\leq l\leq {p_x}$. From left to right, the number of design points are $n=20$, $n=40$ and $n=80$, respectively. The designs are generated from the maximin LHD on $[0,1]^{p_x}$. In the upper panels, the blue solid curves are the density of the JR prior and the red dashed curves are the density of the reference prior (up to the normalizing constants) with $p_x=1$. The densities of the JR prior for ${p_x}=2$ are graphed in the lower panels.  For all the panels, we assume $\mathbf H=\mathbf 1_n$.     \vspace{-.1in}}
\label{fig:prior_comparison_1dim}
\end{figure}


 In  Figure~\ref{fig:prior_comparison_1dim}, the densities of the JR prior and  reference prior with ${p_x}=1$ are graphed in the upper panels. The JR prior matches the reference prior reasonably well. When the number of observations increases, the mass of the JR prior moves to the domain with the larger values of $\bm \xi$, preventing overwhelmingly large  correlation. The densities of the JR prior are graphed in the lower panels with ${p_x}=2$. When the dimension of inputs increases, the mass of the JR prior moves to the domain with the smaller values of $\bm \xi$, preventing the covariance matrix from being estimated to be diagonal. Both features are important for avoiding the degenerate cases discussed in Section \ref{subsec:GP_emulator}.
   
Furthermore, with $a=1/2-{p_x}$, the tail  of the JR prior decreases slightly faster than the reference prior when $\xi\to -\infty$ shown in Figure~\ref{fig:prior_comparison_1dim}.  This is helpful for the identification of the calibration parameters, an example of which is given in Section~\ref{subsec:calibration_numerical}.

\subsection{Emulation}
\label{subsec:emulation_JR_prior}
In this subsection, we discuss parameter estimation with the JR prior in a GaSP emulator introduced in Section \ref{subsec:GP_emulator}. As computer models are often deterministic, the JR prior  has the following form 
\begin{equation}
\pi^{JR}(\beta_1,...,\beta_{p_x})=C_0\left(\sum^{{p_x}}_{l=1} C_l \beta_l \right)^{a} \exp\left\{-b\left(\sum^{p_x}_{l=1} C_l\beta_l \right)\right\}_,
\label{equ:JRprior_no_nugget}
\end{equation}
where $C_0=\frac{({p_x}-1)! b^{a+{p_x}}\prod^{{p_x}}_{l=1}C_l }{\Gamma(a+{p_x}) }$. The JR prior in (\ref{equ:JRprior_no_nugget}) is a special case of (\ref{equ:JRprior}) with $\eta=0$, so the properties of the prior discussed in Section~\ref{subsec:JR_Calibration} can be easily extended to this scenario.


 One important feature of the reference  prior is that the marginal posterior mode estimation is robust under the $\bm \gamma$ and $\bm \xi$ parameterization. Here we show a similar result when the JR prior is used to replace the reference prior in the maximum marginal posterior posterior estimation in (\ref{equ:est_gamma}). 

\begin{thm}(Robust estimation of the JR prior.)
Assume the JR prior in (\ref{equ:JRprior_no_nugget}) with $b>0$ and $C_l>0$.
\begin{itemize}
\item Under the parameterization of the range parameter $\bm \gamma$ and the log inverse range parameter $\bm \xi$, the marginal posterior mode estimation with the JR prior is robust if $a>-{p_x}$,
\item Under the parameterization of the inverse range parameter $\bm \beta$, the marginal posterior mode estimation with the JR prior is robust if $a>0$.
\end{itemize}
\label{thm:robust_est_JR}
\end{thm}
\begin{proof}
By Lemma 3.3 in \cite{Gu2018robustness}, the marginal likelihood  $L(\bm \gamma| \mathbf y) \leq O(1)$ if $\gamma_l \to \infty$ for all $l$, or $\gamma_l \to 0$ for any $l$, $l=1,...,{p_x}$. The results follow from the fact that the density of the prior is zero at the two limits.  
\end{proof}


Note that the marginal posterior mode with the reference prior under the parameterization of the inverse range parameter $\bm \beta$ is not robust. Surprisingly, the marginal posterior mode with the reference prior will always be at $\mathbf {\hat R}=\mathbf 1_n \mathbf 1^T_n$ under $\bm \beta$ parameterization, and should  clearly be avoided (\cite{Gu2018robustness}). By Theorem~\ref{thm:robust_est_JR}, the marginal posterior mode estimation with the JR prior is robust under the $\bm \beta$ parameterization if $a>0$. It has some added advantages for variable selection, as the posterior is positive if any $\beta_l=0$ given in the following remark.


\begin{remark} (Tail rates.)
Assume the JR prior in (\ref{equ:JRprior_no_nugget}) with $a>0$, $b>0$ and $C_l>0$. Here $\bm \beta_{E}$ denotes the vector of $\beta_l$ for all $l \in E$, $E \subset \{1,2,...,{p_x}\}$.

\item[(i.)] When $\bm \beta_{E} \to \bm \infty$,  the natural logarithm of the JR prior approximately decreases linearly with the rate  $-b\sum_{l \in E}\beta_l$.

\item[(ii.)] When  $\beta_l \to 0$ for all $l=1,...,{p_x}$, the natural logarithm of the JR prior  decreases at the rate of $a \log(\sum^{{p_x}}_{l=1} C_l \beta_l) $.

\item[(iii.)] When  $\bm \beta_{E}\to \mathbf 0$ and $\#E<{p_x}$, $\pi^{JR}(\beta_1,...,\beta_{p_x}) $ is finite and positive.
\label{remark:tailrates_JRprior}
\end{remark}

The first and second parts of the jointly robust prior match the exponential and polynomial tail decaying rates of the reference prior discussed in Section~\ref{subsec:GP_emulator}.  The third part is an improvement, which allows the identification of inert inputs by the marginal posterior mode with the jointly robust prior, discussed more in the next section. Note that the third part only holds for the parameter estimation under the parameterization by the inverse range parameter $\bm \beta$, while the JR prior loses such property under the parameterization by the other parameterizations, e.g. $\bm \gamma$ and $\bm \xi$. Thus we propose the following marginal posterior mode estimation with the reference prior
 \begin{equation}
 ({\hat \beta}_1, \ldots {\hat \beta}_{p_x})=  \mathop{\argmax }\limits_{ \beta_1,\ldots, \beta_{p_x}} \left\{ L( { \beta_1}, \ldots, {\beta_{p_x}} \mid \mathbf{y} )  \, \pi^{JR}({\beta_1}, \ldots, {\beta_{p_x}})\right\}.
 \label{equ:est_mode_JR_prior}
 \end{equation}
where $\pi^{JR}(\cdot)$ is the JR prior in (\ref{equ:JRprior_no_nugget}), with $a>0$, $b>0$ and $C_l>0$. Here we use the same default prior parameters $b=1$ and $C_l= n^{-1/{p_x}}|x^{max}_{l}-x^{min}_{l}| $ for the reasons discussed in Section \ref{subsec:JR_Calibration}. $a=1/5$ is implemented in  the RobustGaSP Package as a default choice for emulation (\cite{gu2018robustgasp}).

\section{Variable selection and sensitivity analysis}
\label{sec:variable_selection}

 
This section discusses the issue for identifying inert inputs of computer models. We first introduce a computationally feasible approach of identifying the inert inputs by the JR prior and then discuss the sensitivity analysis approach.


\subsection{Identifying the inert inputs by the JR prior}
\label{subsec:JR_selection}


Assume the GaSP emulator is used to model the computer model output $f^M(\mathbf x)$. W.l.o.g., we assume the input only appears in the covariance function in (\ref{equ:correlation}).  
Variable selection in this context was studied in \cite{schonlau2006screening,linkletter2006variable,savitsky2011variable}. In \cite{schonlau2006screening}, the variable is selected one by one through a screening algorithm with the functional analysis of the variance, while the number of models to be computed is at the order of $p^2$. In \cite{linkletter2006variable}, the size of the range parameters is used as an indicator to decide whether the input is influential and the full posteriors are sampled from a Metropolis Hasting algorithm.  In \cite{savitsky2011variable}, a spike and slab prior is used for the transformation of the range parameters. However, the difficulty with the model selection strategy comes from the computational burden, as  the model space is $2^{p_x}$, and each evaluation of the likelihood requires $O(n^3)$ flops for the inversion of the covariance matrix.





Note that the difficulty of the variable selection in this context comes from the fact that no closed-form marginal likelihood is available. However, when using the product correlation function in (\ref{equ:correlation}), the hope is that, for an inert input $l$,  $\hat {\mathbf R}_l= \mathbf 1_n \mathbf 1_n ^T$, in which it will not affect the correlation matrix $\mathbf R$. Using posterior mode estimation with reference prior, this would happen if ${\hat \gamma_l \to \infty}$. However, as shown in the following lemma, marginal posterior mode estimation with robust parameterizations utilizing the reference prior  cannot identify inert inputs. The proof follows directly from the tail rate computed in Lemma 4.1. and Lemma 4.2. in \cite{Gu2018robustness}.

\begin{lemma}
The marginal posterior of range parameters $\bm \gamma$ (or the logarithm of inverse range parameters $\bm \xi$)  goes to $0$ if some, but not all, $\gamma_l \to \infty$ (or $\xi_l \to 0$), $l=1, \cdots, {p_x}$, for both the power exponential and Mat\'ern correlation functions, when the reference prior in (\ref{equ:refprior}) is used.
\label{lemma:ref_variable_selection}
\end{lemma}




According to Lemma~\ref{lemma:ref_variable_selection}, the marginal posterior mode with two parameterizations will never appear at  $\hat {\mathbf R}_l = \mathbf 1_n \mathbf 1_n ^T$ for any $l$, as the posterior density is $0$ if  $\hat {\mathbf R}_l = \mathbf 1_n \mathbf 1_n ^T$ for some but not all $l$. The identifiability of inert inputs with the posterior mode estimation, however, requires the posterior density is positive when  $\mathbf R_l = \mathbf 1_n\mathbf 1^T_n$ for some but not all $l$ (otherwise it is not a robust parameter estimation).  Other transformation of the reference prior is also less likely to both maintain the robustness parametrization and  identify  inert inputs. Such difficulties could lead to inferior prediction results when some inert inputs are present in computer models.


 
 Luckily, the marginal posterior mode with the JR prior is positive when $\hat {\mathbf R}_l= \mathbf 1_n \mathbf 1_n ^T$ for some but not all $l$, stated in the following lemma.  

\begin{lemma}
The marginal posterior of inverse range parameters $\bm \beta$ is positive if $\mathbf 1_n \notin \mathcal C(\mathbf H)$ and some, but not all, $\beta_l \to 0$, $l=1, \cdots, {p_x}$, for both the power exponential and Mat\'ern correlation functions, when the  JR prior in (\ref{equ:JRprior_no_nugget}) with $a>0$, $b>0$ and $C_l>0$.
\label{lemma:JR_inert_inputs}
\end{lemma}
The proof of Lemma~\ref{lemma:JR_inert_inputs} follows from the tail rate of the marginal likelihood of the GaSP model (Lemma 3.1 and Lemma 4.1 in \cite{gu2018robustgasp}) and the tail rate of the JR prior in Remark \ref{remark:tailrates_JRprior}. 
When $\beta_l=0$, the  $l$th input is not in the covariance in GaSP model. In practice, exact zero estimation is not likely to be obtained. Thus, we  use the {normalized inverse range parameters} as an indicator of the importance of each  input
\begin{equation}
 \hat P_l =\frac{C_l \hat \beta_l}{\sum^{{p_x}}_{i=1}C_i\hat \beta_i},
 \label{equ:P_l}
 \end{equation} 
 where  $(\hat \beta_1,...,\hat \beta_{p_x})$ are estimated in Equation (\ref{equ:est_mode_JR_prior}). The involvement of $C_l$ is to take  the scale of  different inputs into account. The  part $\sum^{{p_x}}_{l=1}C_l\hat \beta_l$ in the denominator is the overall size of the estimation and the $C_l \hat \beta_l$ is the contribution by the $l$th input. 

 The size of the inverse range parameters has been used to infer which the input is inert, but with a different prior (\cite{linkletter2006variable}). However, the jointly robust prior yields better results, as compared in the Section~\ref{subsec:numerical_VS}.

One may use a certain threshold of the normalized inverse range parameters to predict whether the input is inert or not, i.e. 
\begin{equation}
\hat P_l\leq {p_0}/{{p_x}},
\label{equ:P_l_p_0}
\end{equation}
 where $p_0$ may be chosen as a constant between 0 to 1.  Such values could also depend on the number of observations, dimension of the inputs and expected number of inputs to be chosen. However, because all inputs affect the outputs in a computer model, the threshold might be less important and can be chosen based on the scientific goal.
We do not try to present a method for the full model selection, as each computation of the likelihood can be expensive. The point, here, is that the computation of the  $P_l$ does not take any extra computation (as the posterior modes are typically needed for building a GaSP model), and can serve as a indicator to tell an input is inert or not.






\subsection{Sensitivity analysis}
\label{subsec:sensitivity_analysis}
{{Sensitivity analysis}} in computer model concerns with the problem of learning how changes of inputs affect the outputs.  
The inputs, in the computer model, typically associate with a distribution $\pi(\mathbf x)$, reflecting the belief of the input values. One of the main goal related to the  sensitivity analysis is to identify how much a set of inputs influence the variability of outputs, which is studied through the functional analysis of the variance ({functional ANOVA}). 

It is possible to decompose a function $f^M(\cdot)$ as follows (\cite{hoeffding1948class})
\[f^M(\mathbf x) =z_0+\sum^{{p_x}}_{i=1} z_i(x_i)+\sum^{{p_x}}_{i<j}z_{ij}(\mathbf x_{i,j})+...+z_{12...{p_x}}(\mathbf x),\]
where $\mathbf x_{i,j}=(x_i,x_j)$ and $\mathbf x=(x_1,...,x_{p_x})$.  One can obtain these element functions by taking expectation on $\mathbf x$: $z_0=\E[f^M(\mathbf x)], $ $z_i(x_i)=\E[f^M(\mathbf x)|x_i]-z_0$, $z_{ij}(\mathbf x_{i,j})=\E[f^M(\mathbf x)|\mathbf x_{i,j}]-z_0-z_i-z_j, $ and so on.
Here $z_i(x_i)$ is often referred as the main effect and $z_{i,j}(\mathbf x_{i,j})$ is referred as the second order effect.  

The variance of the function can be decomposed (\cite{efron1981jackknife})
\[\Var[f^M(\mathbf x)]= \sum^{{p_x}}_{i=1}{W_i}+\sum^{{p_x}}_{i<j}{W_{ij}} +...+W_{12...{p_x}}, \]
where $W_i=\Var[\E[f^M(\mathbf x)|x_i]]=\Var[z_i(x_i)]$, $W_{ij}=\Var[\E[f^M(\mathbf x)|\mathbf x_{i,j}]]-W_i-W_j $.   Two principal measures called the main effect index and the total effect index were defined as (\cite{sobol1990sensitivity})
\begin{align*}
S_i&=V_i/\Var[f^M(\mathbf x)], \\
S_{T_i}&=V_{T_i}/\Var[f^M(\mathbf x)],
\end{align*}
where $V_i=W_i=\Var[\E[f^M(\mathbf x)|x_i]]$ and $V_{T_i}=\Var[f^M(\mathbf x)]-\Var[\E(f^M(\mathbf x)\mid \mathbf x_{-i})]$. $S_i$ is referred as the main effect index of $x_i$ and $S_{T_i}$ is referred as the total effect index of $x_i$.


As pointed out in \cite{oakley2004probabilistic}, $S_i$ has very clear interpretation. If we were to know the real value of the $i$th input, denoted as $x^{r}_i$, the uncertainty left  is thus $\Var[f^M(\mathbf x)\mid x_i= x^{r}_i]$, and the decrease of the uncertainty is  $\Var[f^M(\mathbf x)]-\Var[f^M(\mathbf x)\mid x_i=x^{r}_i]$. Since we do not know $x_i$, it is common to take the expectation. Consequently, the decrease of the variance is then $\Var[\E[f^M(\mathbf x)\mid x_i]]=V_i$. It means if we were able to select one input to explore its true value, we will select $x_i$ that maximizes $V_i$.

However, if one were able to select two inputs to explore, the answer is {not} to select the largest main effect index, but to select the largest $V_{i,j}$ as follows
\[V_{i,j}=\Var[\E(f^M(\mathbf x)\mid \mathbf x_{i,j})]=\Var[z_i(x_i)+z_j(x_j)+z_{ij}(\mathbf x_{ij})].\]
Thus,  many higher order indices are needed to compute if one are interested in exploring more than the first few influential inputs. Main effect indices may serve as an approximation, and  they are frequently used due to the computational reason.


When  $f^M(\mathbf x)$ and $\pi(\mathbf x)$ have simple  forms,  the main effect indices and higher order indices may be computed explicitly. However, these indices generally do not have a closed form expression, thus the numerical estimation of these indices becomes important. Monte Carlo methods are proposed to evaluate these indices (\cite{sobol1990sensitivity,sobol2001global}).

The shortage of the Monte Carlo method is that lots of computer model runs are often needed for numerically estimation, which  is unrealistic when the computer model is slow. One approach that significantly reduces the number of evaluation of the functions is discussed in \cite{oakley2004probabilistic}. The idea is to use a small number of runs to fit the GaSP emulator and  the posterior predictive distribution is used to replace the computer model outputs. The estimation of the indices can be implemented based on the emulator built on only very small number of runs from the computer model. We compare with these methods in Section~\ref{sec:numerical}.

\section{Numerical Study}
\label{sec:numerical}
\subsection{Emulation}
\label{subsec:emulation}
We numerically compare the predictive performance of GaSP emulator using the  marginal posterior mode estimation with the JR prior and reference prior. For the reference prior, we choose the log inverse range parameterization, $ \xi_l=\log(1/\gamma_l)$, because it is both robust and has empirically better predictive performance than the $\gamma_l$ parameterization (\cite{Gu2018robustness}).  Both methods are implemented in the RobustGaSP package (\cite{gu2018robustgasp}). The  Mat{\'e}rn correlation with $\alpha_l=5/2$ is used and a constant mean basis function is assumed for all cases (i.e. $h(\mathbf x)=1$). Also included are the results from  the DiceKriging Package   (\cite{roustant2012dicekriging}) with the same correlation function and mean basis function. 

In each experiment, we use $n$ inputs to construct the GaSP emulator, where $n$ is typically chosen to be around $10{p_x}$, and then record the out-of-sample normalized root of mean square error (NRMSE) of $n^*=10,000$  held-out outputs. We repeat the experiments for $N=200$ random designs, generated from the maximin LHD (\cite{lhspackage}), and report the  average normalized root of mean square error (Avg-NRMSE):
\begin{eqnarray}
\label{equ:NRMSE}
\begin{split}
\mbox{NRMSE}_j&= \sqrt{  {\sum\limits_{i = 1}^{{n^{*}}} {{(y(\mathbf x^{*}_{ij})  - \hat y(\mathbf x^{*}_{ij})   )}^2}}/   \sum\limits_{i = 1}^{{n^{*}}} {{(y(\mathbf x^{*}_{ij})  - \bar {\mathbf y}_j   )}^2}}, \\
\mbox{Avg-NRMSE}&=\frac{1}{N}\sum^N_{j=1}\mbox{NRMSE}_j, \end{split}
\end{eqnarray}
with $\mathbf x^{*}_{ij}$ being the $i${th} held-out input in the $j${th} experiment, $ \hat y(\mathbf x^{*}_{ij})$ being its prediction and  $\bar {\mathbf y}_j$ being the mean of the observed output in the $j$th experiment, $j=1,...,N$.  

We  test the following  functions (implemented in \cite{simulationlib}).
\begin{example}
 \begin{itemize}
  \item[i.]    $Y=[X_2-5.1X^2_1/(4\pi^2)+5X_1/\pi-6]^2+10[1-1/(8\pi)]\cos(X_1)+10$ where $X_i \in [0,1]$, for  $i=1,2$.

  \item[ii.]   $Y=4(X_1-2+8X_2-8X^2_2)^2+(3-4X_2)^2+16\sqrt{X_3+1}(2X_3-1)^2$, where $X_i \in [0,1]$, for $i=1,2,3$.

    \item[iii.]     $Y=2\exp\{\sin[0.9^8(X_1+0.48)^8]\} +X_2X_3+X_4$, where  $X_i \in [0,1)$, for $i=1,2,3,4$.

   \item[iv.]      $Y=10\sin(\pi X_1X_2)+20(X_3-0.5)^2+10X_4+5X_5$, where  $X_i \in [0,1]$, for $i=1,2,3,4,5$.

   \item[v.]   $Y=\frac{2\pi X_3(X_4-X_6)}{ \ln(X_2/X_1)\{1+{2X_7X_3}/ [ \ln(X_2/X_1)X_1^2X_8  ] +{X_3}/{X_5} \}},$
 where $  X_1 \in [0.05, 0.15]$, $X_2 \in [100, 50000]$, $X_3 \in [63070, 115600]$ , $X_4 \in [990, 1110]$, $X_5 \in [63.1, 116]$, $X_6 \in  [700, 820]$, $X_7 \in [1120, 1680]$ and $X_8 \in [9855, 12045]$ are the 8 inputs. 
   \end{itemize}
   \label{eg:emulation}
   \end{example}

\begin{table}[t]
\begin{center}
\begin{tabular}{llll}
  \hline
                     & Robust GaSP $\bm \xi$ & JR prior  &DiceKriging  \\

  \hline
  case i               & $.028$ ($.15$ s) &$.028$ ($.054$ s) &$.063$ ($.029$ s) \\
  case ii                & ${.011}$ ($.53$ s) &$.011$ ($.10$ s) &$.061$ ($.04$ s)  \\
  case iii              & $.059$ ($1.1$ s) &$.051$ ($.15$ s) &$.21$ ($.059$ s)   \\
  case iv         & ${.018}$ ($3.3$ s) & $.018$ ($.37$ s)& $.10$ ($.074$ s)   \\
  case v         & ${.0093}$ ($20$ s)& $.0094$ ($1.4$ s)& $.094$  ($.43$ s)   \\

    \hline
\end{tabular}
\end{center}
   \caption{Avg-NRMSE and average computational time in seconds for parameter estimation in the bracket of the three estimation procedures for the five experimental functions in Example~\ref{eg:emulation}. From the upper row to the lower row, the sample size is $n=30, 40, 50, 60$ and  $80$ for these five cases, respectively.  }
   \label{tab:Avg-NRMSE}
\end{table}
The Avg-NRMSEs of three methods of the five testing functions in Example~\ref{eg:emulation} are shown in Table~\ref{tab:Avg-NRMSE}. The Avg-NRMSE of the methods with the reference prior and JR prior is similar, whereas the computational time of the JR prior is smaller, as the closed-form derivatives of the JR prior are known. The DiceKriging is the fastest method among three but the predictions are not as good as the  robust methods.

The difference of the computational time for searching the posterior mode between the JR prior and reference prior becomes larger (shown in Figure 4.4. in \cite{Gu2016thesis}). This is because the numerical derivative of the reference prior requires many extra evaluations of the likelihood, each having $O(n^3)$ operations. In comparison, when the JR prior is used, we only need to compute the Cholesky decomposition of the correlation matrix once for the inversion of the covariance matrix in each iteration, because of the closed form derivatives.



\subsection{Variable Selection}
\label{subsec:numerical_VS}

We first study the following example reported in \cite{linkletter2006variable}.
\begin{example}
 i. $Y=0.2 X_1+0.2 X_2+0.2X_3+0.2X_4+\epsilon,  $
where $\epsilon \sim N(0,0.05^2)$ and $X_l \in [0,1]$, $l=1,..,4$. 6 completely noise input variables are also added.

ii. $Y=0.2 X_1+0.2/2 X_2+0.2/4X_3+0.2/8X_4 +0.2/16X_5+0.2/32 X_6+0.2/64X_7+0.2/128X_8+\epsilon,$
where $\epsilon \sim N(0,0.05^2)$ and $X_l \in [0,1]$, $l=1,..,8$. 2 completely noise input variables are also added.
\label{eg:VS_eg1}
\end{example}

The number of design points is $n=54$ and the Gaussian correlation function is assumed for both functions in Example~\ref{eg:VS_eg1}, same as in \cite{linkletter2006variable}.  The functions are linear, however, in the GaSP model, we only use the constant mean function, pretending that we don't know the linear trend of the real function. $N=1,000$ random designs are generated from the maximin LHD design (\cite{lhspackage}).

\begin{figure}[t]
\centering
  \begin{tabular}{cc}
	\includegraphics[height=.4\textwidth,width=.5\textwidth]{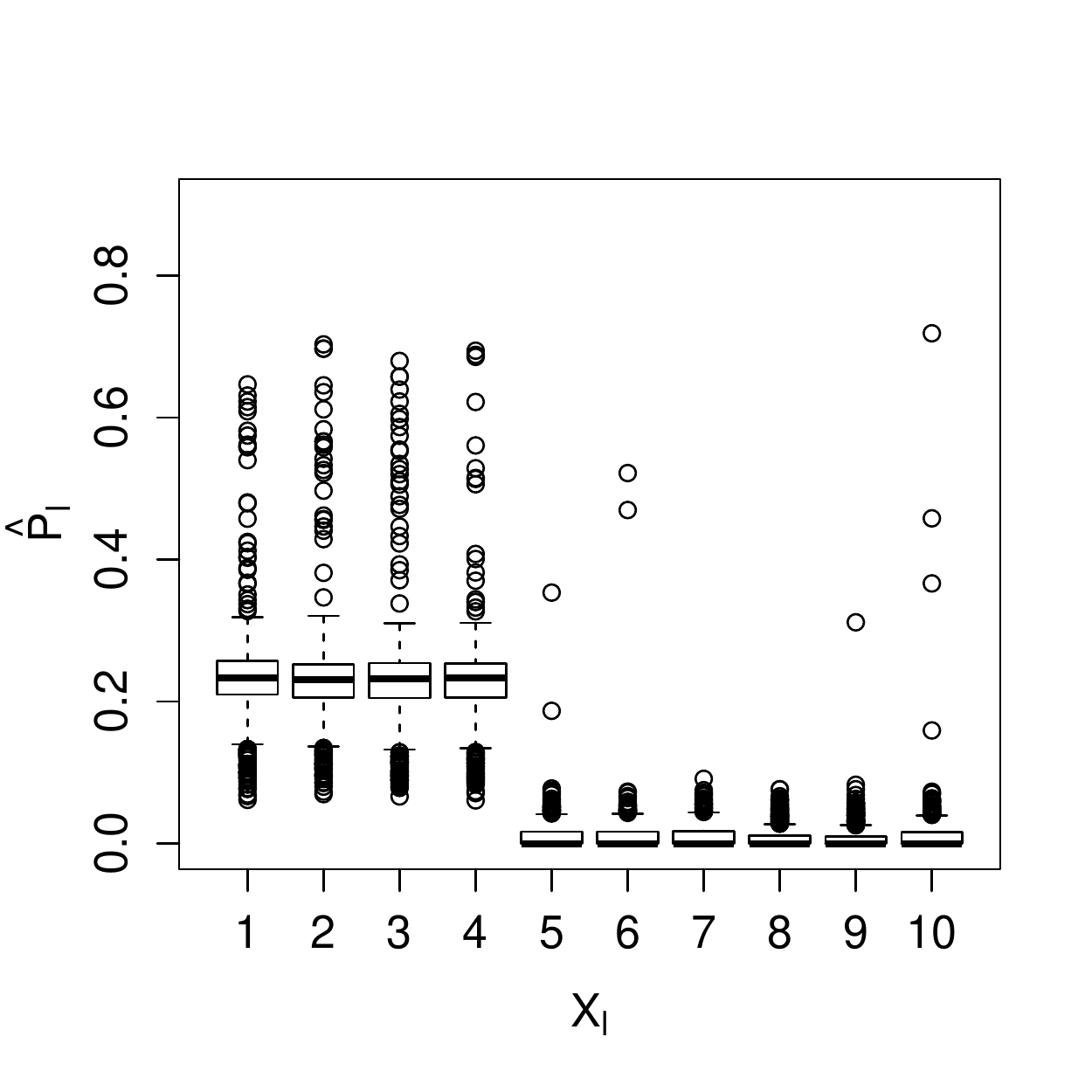}
	\includegraphics[height=.4\textwidth,width=.5\textwidth]{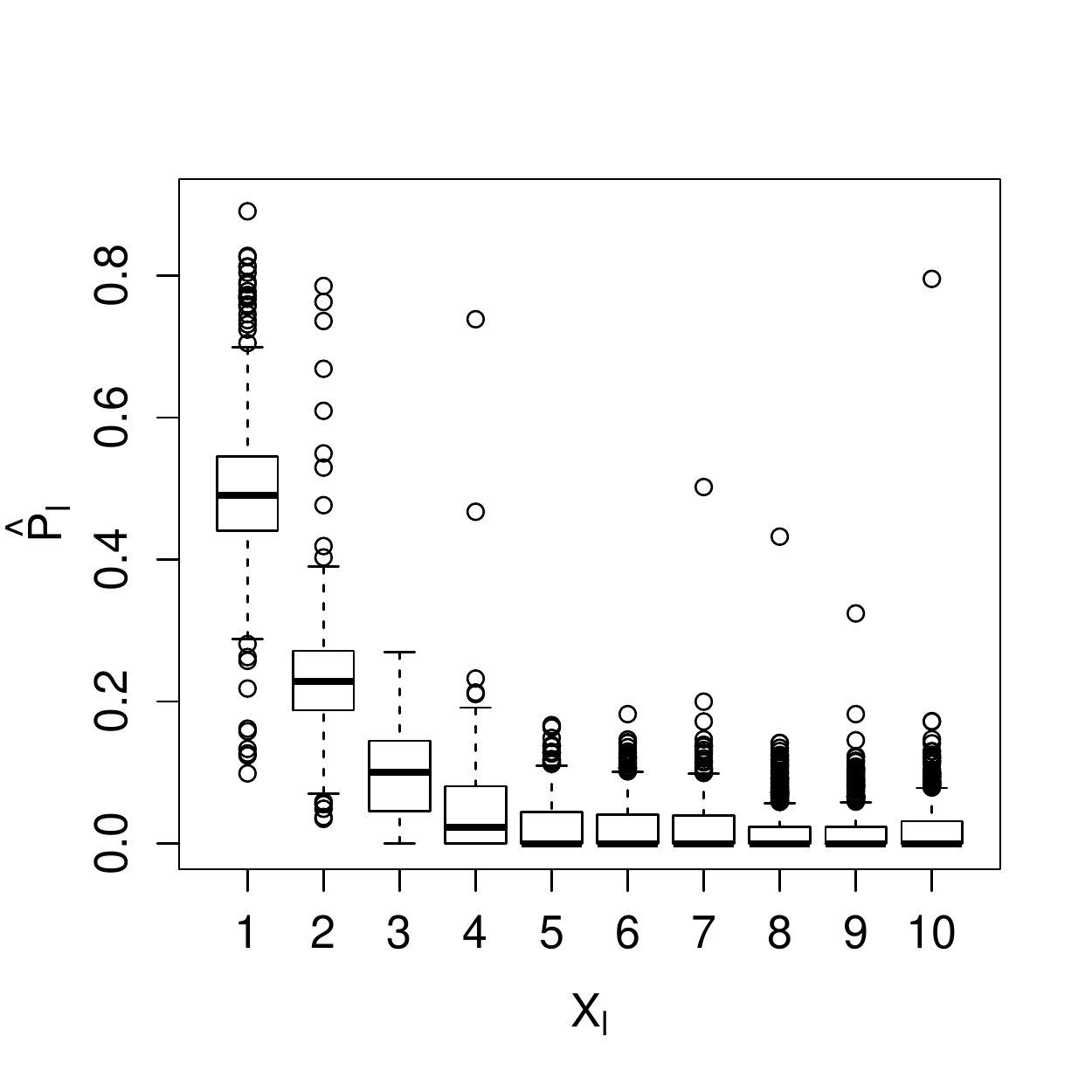}
  \end{tabular}
  \vspace{-.2in}
   \caption{Normalized inverse range parameter $\hat P_l$, $l=1,...,p_x$, for case i (left panel) and case ii (right panel) in Example~\ref{eg:VS_eg1}.  }
\label{fig:david_higdon_2eg}
\end{figure}

The  normalized inverse range parameter $\hat P_l$ are shown in Figure~\ref{fig:david_higdon_2eg}. In the left figure, it is clear that the first four inputs are much more important than the rest of input. Indeed these are 4 signals while the other are noises. The second figure shows that the normalized inverse range parameters can identify the largest 3 to 4 signals. A detailed comparison with results in \cite{linkletter2006variable} are shown in Table~\ref{tab:dheg1_2}. In both cases,  using $\hat P_l$ with the JR prior has smaller false positives and false negatives, compared with the reference distribution variable selection (RDVS) method in \cite{linkletter2006variable}.

\begin{table}[t]
   \caption{Proportion of times each input is identified as influential inputs in Example~\ref{eg:VS_eg1} by the JR prior with different $p_0$ in (\ref{equ:P_l_p_0}) and RDVS method in \cite{linkletter2006variable} with different percentiles (PT). The normalized inverse range parameter $\hat P_l$, $l=1,...,p_x$,  in (\ref{equ:P_l}) with JR prior and different $p_0$ is used to identify the inert inputs.
    }

\small
\begin{tabular}{lcccccccccc}
  \hline
   case i                      & 1&2&3&4&5&6&7&8&9&10  \\
    \hline
  JR prior, $p_0=1$                     &.974&.979&.967&.974 &.003&.004&.001&.002&.002&.006\\
    JR prior, $p_0=.75$                     &.994&.997&.997&.995&.006&.008&.005&.007&.006&.009\\
    JR prior, $p_0=.5$                     &1&1&1&1&.034&.033&.035&.032&.022&.043\\
   \hline
  RDVS, $5$th PT            &.619 &.618& .717 &.631& .030& .034& .021& .074& .051& .051 \\
   RDVS, $10$th PT          &.852 &.855& .910 &.880&.061& .064& .053& .137& .076& .102 \\
    RDVS, $15$th PT          &.947 &.954& .973 &.955&.079& .091& .080& .173& .108& .135 \\
  \hline
   case ii                  & 1&2&3&4&5&6&7&8&9&10  \\
    \hline
  JR prior, $p_0=1$                     &.998&.967&.504&.165 &.074&.050&.041&.038&.026&.041\\
    JR prior, $p_0=.75$                     &1&.982&.633&.267&.126&.119&.093&.089&.072&.111\\
    JR prior, $p_0=.5$                     &1&.993&.736&.383&.217&.214&.196&.165&.157&.187\\
   \hline
  RDVS, $5$th PT            &.679 &.180 &.062 &.025 &.016 &.023 &.017 &.031 &.009 &.036 \\
   RDVS, $10$th PT          &.889 &.379 &.133 &.058 &.034 &.051 &.035 &.067 &.030 &.094 \\
    RDVS, $15$th PT          &.959 &.540 &.217 &.092 &.061 &.098 &.065 &.107 &.063 &.149 \\

    \hline
\end{tabular}
   \label{tab:dheg1_2}
\end{table}

The cut-off value $p_0$ for the JR prior can be hard to define. However, typically all inputs influence the outputs of the computer models.  The task is then not to identify the true signals, but to identify what set of inputs are more important than the others. This seems successful for both functions in Example \ref{eg:VS_eg1}, as the importance of the factors is correctly ordered, shown in Figure \ref{fig:david_higdon_2eg}. 

  In the following Example \ref{eg:manyfunct_noise}, we test  the following four functions (implemented in \cite{simulationlib}) to check whether the method can identify the set of signals.  For the approach with the JR prior, we use the normalized inverse range parameters in (\ref{equ:P_l}) of each input as the index of whether it is a signal. We evaluate the performance by the fraction of times when the smallest index of the signals is larger than the largest index of the noises over $N$ experiments.
  


%

\begin{example}
 \item[i.]   $Y=\frac{1}{6}[(30+5X_1 sin(5X_1))(4+exp(-5X_2))-100]  +\epsilon$, $X_i \in \mbox{Unif} (0,1)$, $i=1,...,7$, $\epsilon \in N(0,0.3^2)$.

 \item[ii.]  $Y=4(X_1-2+8X_2-8X^2_2)^2+(3-4X_2)^2+16\sqrt{X_3+1}(2X_3-1)^2+\epsilon$, $X_i \in \mbox{Unif} (0,1)$, $i=1,...,6$, $\epsilon \in N(0,0.05^2)$.

 \item[iii.]  $Y=\frac{2}{3}\exp(X_1+X_2)-X_4\sin(X_3)+X_3+\epsilon$,  $X_i \in \mbox{Unif} (0,1)$, $i=1,...,8$, $\epsilon \in N(0,0.15^2)$.

 \item[iv.] $Y=10sin(\pi X_1X_2)+20(X_3-0.5)^2+10X_4+5X_5+\epsilon$, $X_i \in \mbox{Unif} (0,1)$, $i=1,...,10$, $\epsilon \in N(0,0.2^2)$.
 \label{eg:manyfunct_noise}
\end{example}

\begin{table}[t]
   \caption{Tested sample size and the fraction of times when the smallest index of the signals is larger than the largest index of the noises recorded in the bracket for Example~\ref{eg:manyfunct_noise} by different methods over $N=200$ experiments. }
\begin{center}
\small
\begin{tabular}{p{.5in}p{.7in}cccc}
  \hline
                    &JR prior & Sobol GaSP &Sobol & Sobol2007 S &Sobol2007 T  \\
  \hline

   Case i &  $20 \, (.960)$   &$20  \,(.970)$ &$130  \,(.835)  $ & $130  \,(.825)$ & $130  \,(.820)$                    \hfill \\
    Case ii &    $ 35  \,(.965)$ &  $40 \,(.290)$   & $20,000  \,(.800)$ &$20,000  \,(.780)$& $600  \,(.820)$          \hfill \\
   Case iii &  $35  \,(.905)$  & $35  \,(.795)$ & $4,000  \,(.695)$  & $4,000  \,(.675)$& $4,000  \,(.765)$ \hfill \\
   Case iv &     $35  \,(.920)$ &  $35  \,(.780)$& $1,000  \,(.810)$  & $1,000  \,(.805)$ & $1,000  \,(.790)$ \hfill \\
   \hline

\end{tabular}
\end{center}
   \label{tab:many_comparison}
\end{table}
The results of Example~\ref{eg:manyfunct_noise} are recorded in Table~\ref{tab:many_comparison}. From the left to right, it records the performance  using the  normalized inverse range parameter, Sobol GaSP (\cite{oakley2004probabilistic,le2014bayesian}), Sobol (\cite{sobol1990sensitivity}), Sobol2007 S and Sobol2007 T (\cite{tarantola2007estimating}).  All Sobol methods are coded in the  Sensitivity package (\cite{pujol2007sensitivity}).

 In Table~\ref{tab:many_comparison}, Sobol method with Monte Carlo method (and its variants) needs much more computer model runs to identify the signals, while Sobol GaSP needs much less runs, consistent with the previous study in  \cite{oakley2004probabilistic}.  The Sobol GaSP is not as good as the  normalized inverse range parameter with the JR prior. One possible reason is that  the Sensitivity package (\cite{pujol2007sensitivity}) utilizes the DiceKriging package (\cite{roustant2012dicekriging}) for the GaSP emulator, which is not as accurate as the robust GaSP emulator  in prediction, discussed in Section~\ref{subsec:emulation}. 

\subsection{Calibration}
\label{subsec:calibration_numerical}
	In this section, we compare the GaSP and S-GaSP calibration for a pedagogic example studied in \cite{bayarri2007framework}.
	\begin{example}
	The sampling model is  $y^F(x)=3.5\exp(-1.7x)+1.5+\epsilon$ with $\epsilon \sim \mathcal N(0,0.3^2)$, and the computer model is $f^M(x,\theta)= 5 \exp(-\theta x)$. Thirty observations are recorded  at 10 different $x_i \in [0,3]$, each with three repeated experiments shown in \cite{bayarri2007framework}.  The goal is to estimate $\theta$ and predict the outputs at $[0,5]$.
	\label{example:calibration}
	 \end{example}
	
  Because the uncertainty of the calibration parameters is important in calibration, sampling from the posterior is typically more preferred than the MLE or posterior mode estimation. The Markov Chain Monte Carlo Algorithm is implemented in RobustCalibration package (\cite{gu2018robustcalibrationpackage}). We compare the GaSP calibration model in (\ref{equ:model_calibration}) with the reference prior and JR prior, based on $S=100,000$ posterior samples with $S_0=20,000$ burn-in samples. As the computer model does not explain the mean of the process, we add a mean discrepancy term to the computer model (i.e. $h(\mathbf x)=1$) for all the models we considered. The mean discrepancy is treated as a part of the computer model because of its interpretability.  	
		\begin{figure}[t]
		\centering
		\begin{tabular}{c}
			\includegraphics[height=.45\textwidth,width=1\textwidth]{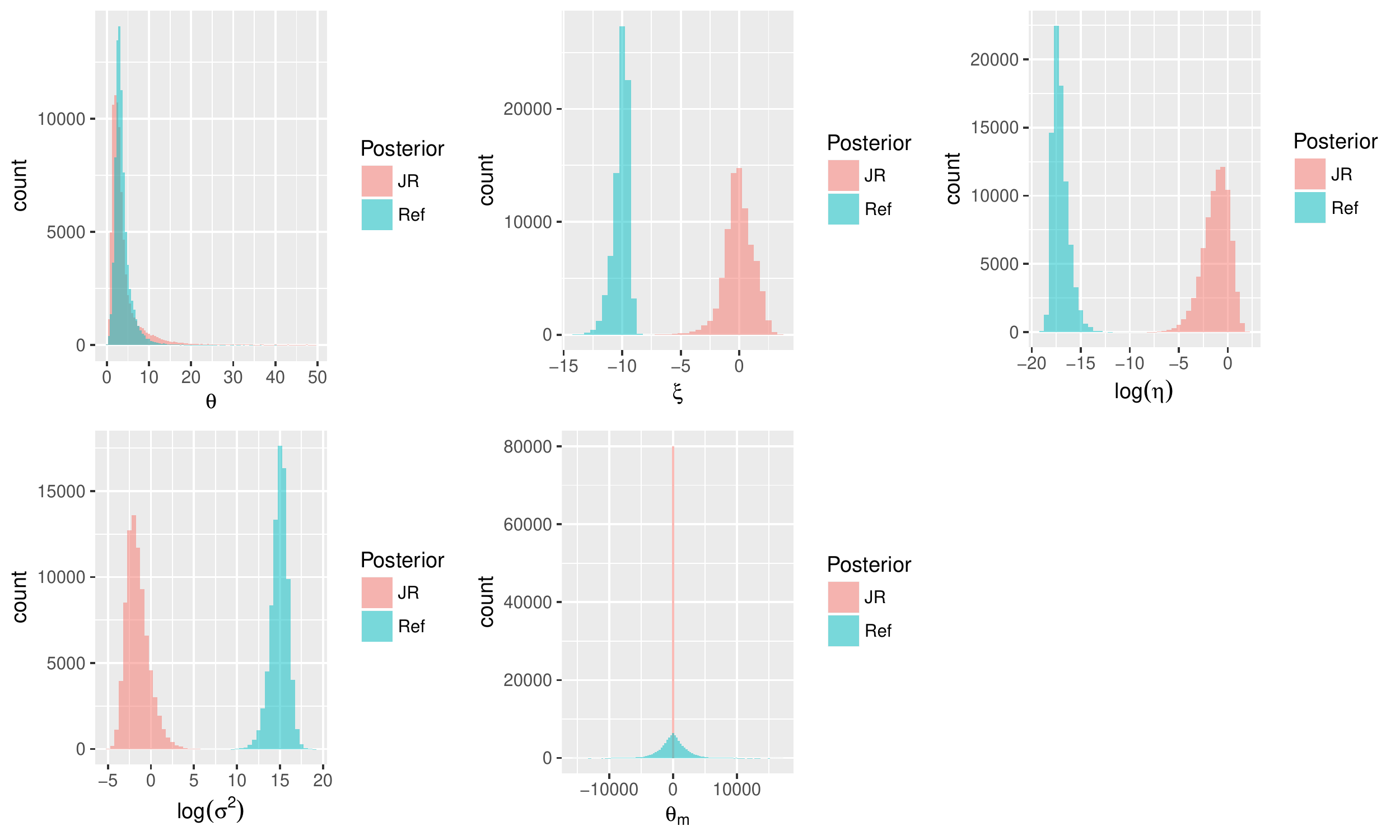}	 
		\end{tabular}
		\caption{Histograms of the posterior samples of the parameters with the JR prior (red boxes) and the reference prior (green boxes) for the Example~\ref{example:calibration}.  }
		\label{fig:eg_posterior_calibration}
	\end{figure}
	
	The posterior samples of the parameters with the reference prior and JR prior are graphed in the Figure~\ref{fig:eg_posterior_calibration}. Since the reference prior has a flatter tail when the log inverse range parameter $\xi \to 0$, the posterior of $\xi$ with the reference prior is much smaller than the one with the JR prior, meaning that the correlation is estimated to be much stronger.  The large correlation by the reference prior leads to the large values of the posterior samples of the variance parameter, shown in the left panel in the second row in Figure~\ref{fig:eg_posterior_calibration}, and consequently, the posterior samples of the mean parameter $\theta_m$ spread widely over $[-2\times 10^4, 2\times 10^4]$, shown in the last panel in Figure~\ref{fig:eg_posterior_calibration}. In comparison, the posterior mean parameter with the JR prior is much more concentrated, because the tail of the density of JR prior is slightly steeper  when the log inverse range parameter $\xi \to 0$, preventing the correlation from being estimated to be too large.  
	
	




	
	To see the predictive performance of calibration, We test on $n^*=200$ held-out outcomes at $x^*_i$ equally spaced at $[0,5]$ based on the predictive NRMSE in (\ref{equ:NRMSE}) and the  following additional two criteria
	\begin{eqnarray}
	{\rm P_{CI}(95\%)} &=& \frac{1}{n^{*}}\sum\limits_{i = 1}^{n^{*}} 1\{y^R(x^{*}_i)\in {\rm CI}_{i}(95\% )\}, \nonumber\\
	{\rm L_{CI}(95\%)} &=&\frac{1}{n^{*}} \sum\limits_{i = 1}^{n^{*}} \Length\{\rm CI_{i}(95\%)\}, \nonumber
	\end{eqnarray}
	where ${\rm CI}_{i}(95\% )$ is the $95\%$ posterior credible interval of the reality; and ${\rm L_{CI}(95\%)}$ is the average length of the $95\%$ posterior credible interval. For the results by the reference prior and JR prior,  NRMSE is calculated for two scenarios. In the first scenario, only the calibrated computer model is used for prediction. Both the calibrated computer model and discrepancy function are used in the second scenario. An efficient method should {{have}} relatively low predictive  NRMSE for both scenarios, ${\rm P_{CI}(95\%)}$  close to the $95\%$ nominal level and short average credible interval lengths.

		\begin{figure}[t]
		\centering
		\begin{tabular}{ccc}
			\includegraphics[scale=.285]{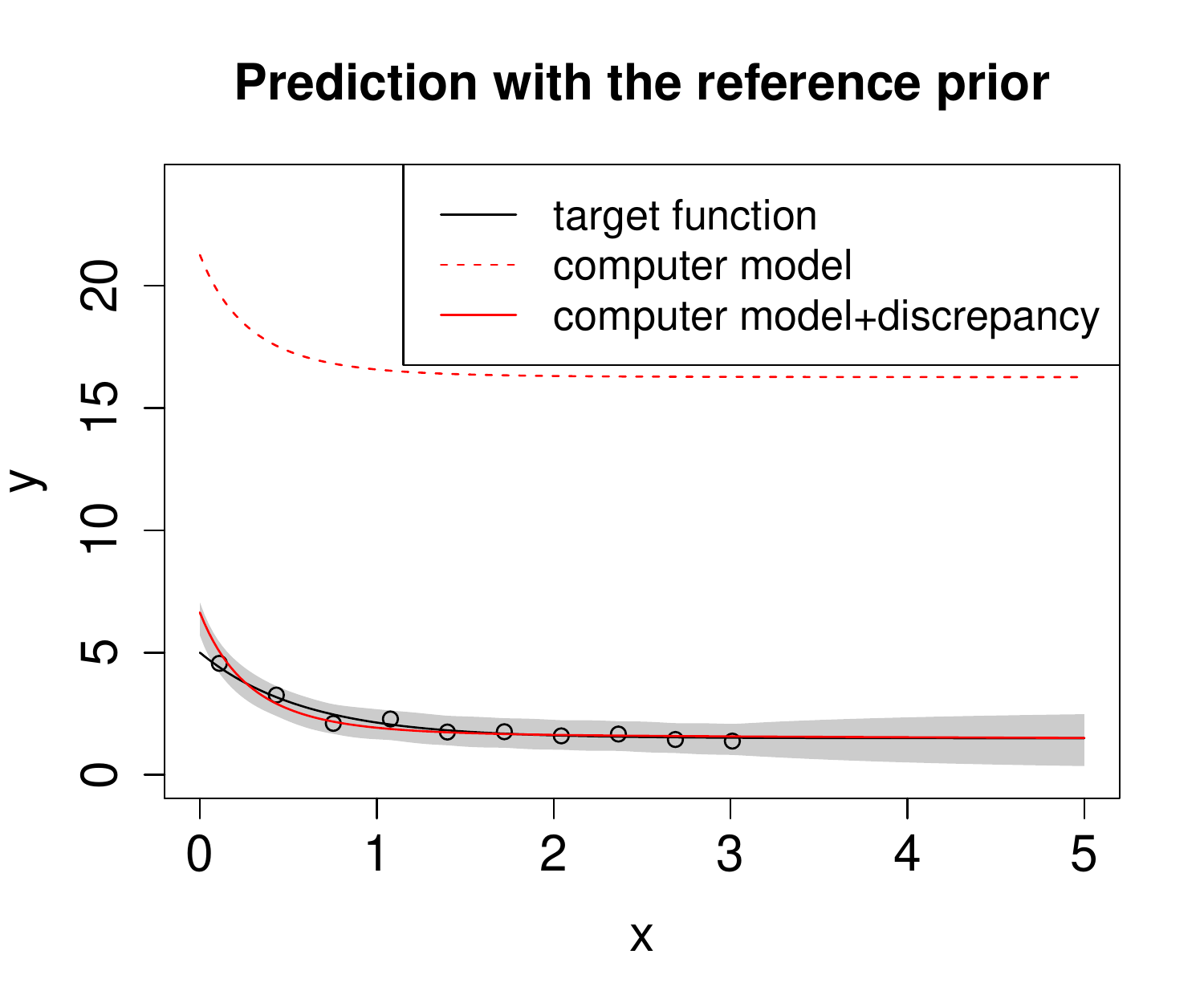}	 
		\includegraphics[scale=.285]{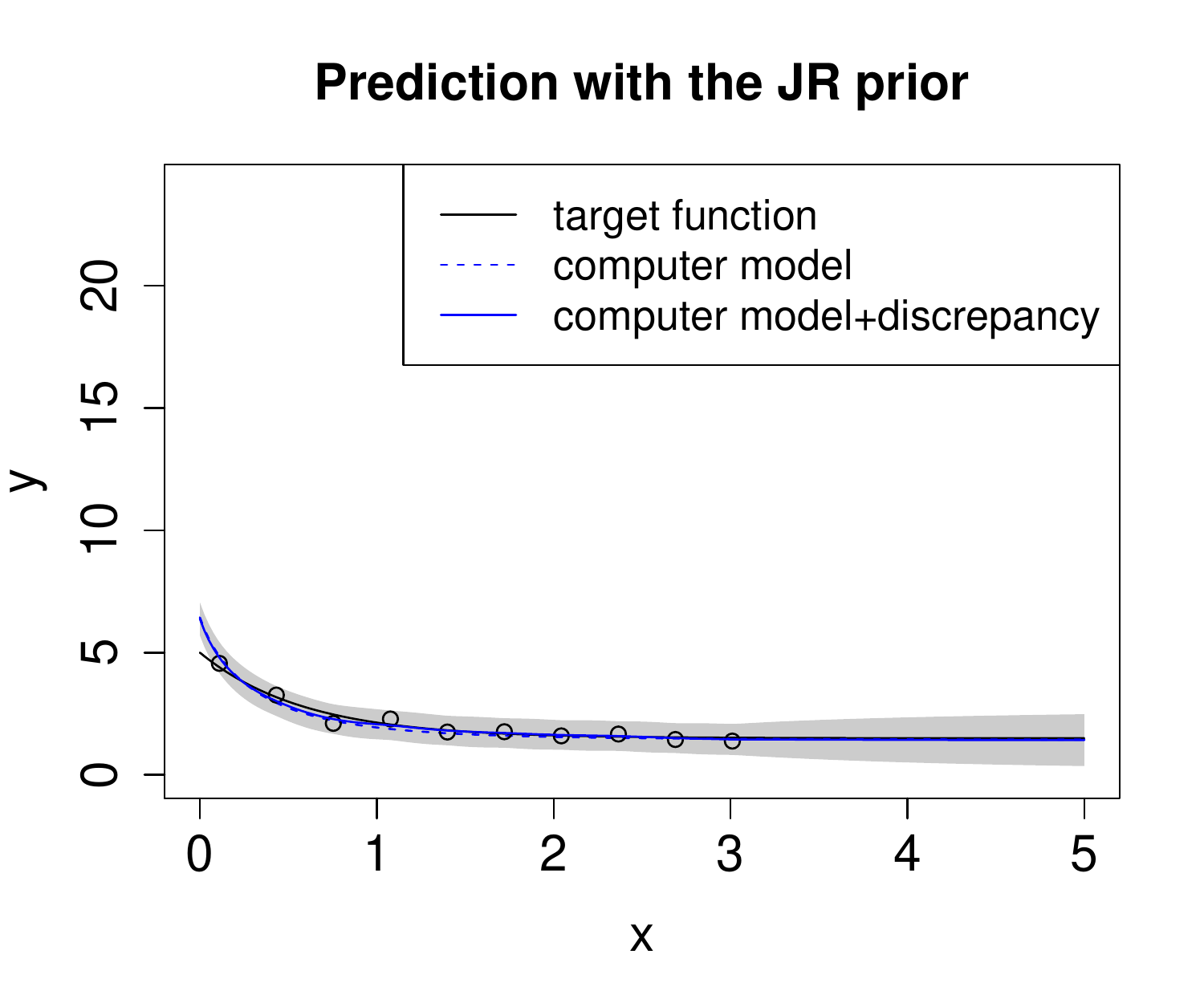}	 
		\includegraphics[scale=.285]{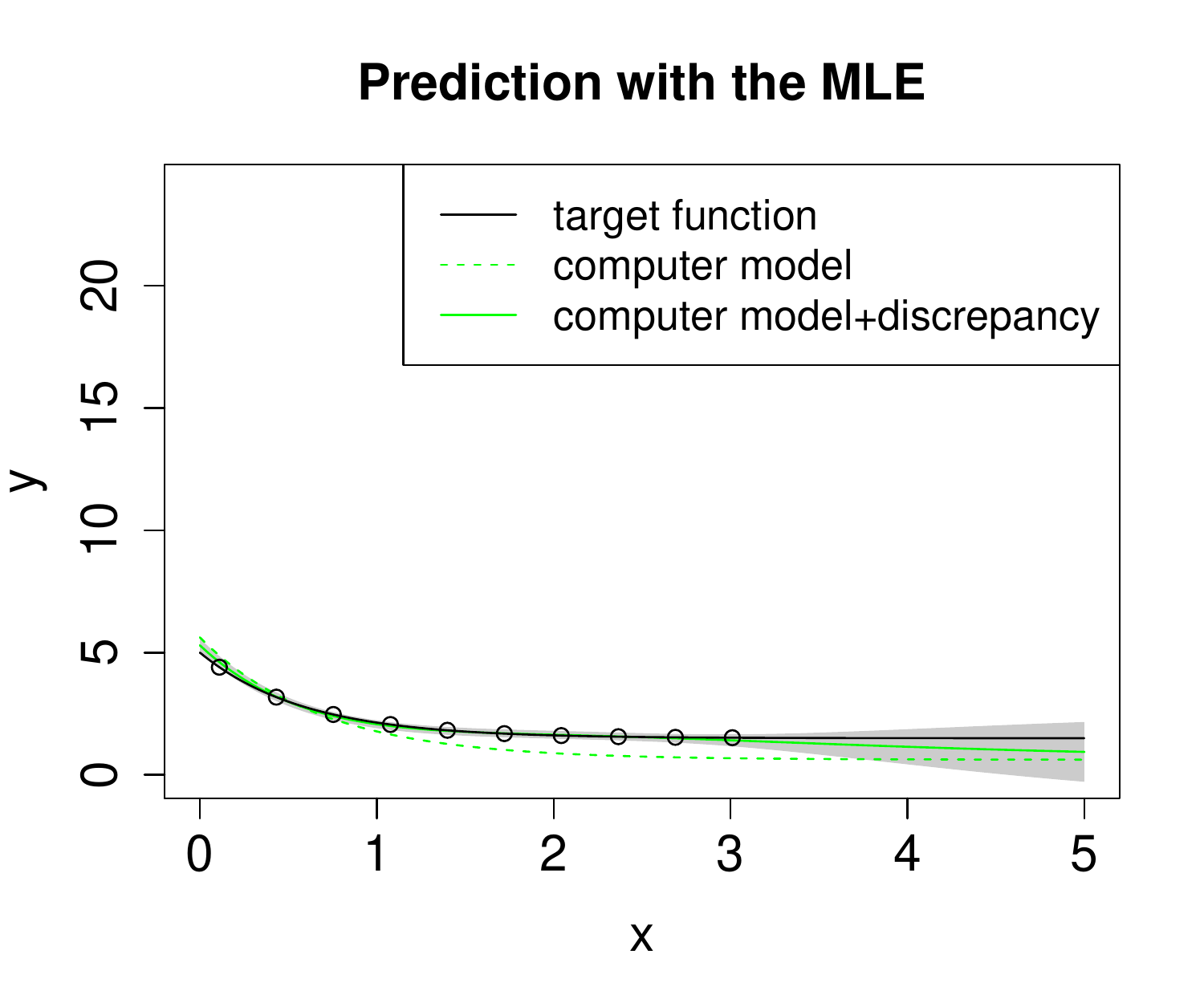}	 	
		\end{tabular}
		\caption{Out of sample prediction for the Example~\ref{example:calibration} with the reference prior (left panel), JR prior (middle panel) and MLE (right panel). The target function is graphed as black {solid} curves and  30 observations are plotted as circles. The colored solid curves are the predictive mean using the calibrated computer model and discrepancy function, while the colored dashed curves are the predictive mean using only the calibrated computer model. The shaded area is the 95\% predictive credible interval for the target function. In the middle  panel, the dash curve and the solid curve almost overlap. }
		\label{fig:eg_prediction_calibration}
	\end{figure}

	\begin{table}[t]
		\caption{\label{tab:predicition_calibration} NRMSE,  ${\rm P_{CI}(95\%)} $ and ${\rm L_{CI}(95\%)} $ by the GaSP calibration with the reference prior, JR prior and MLE for Example~\ref{example:calibration}.}
		\centering
		\begin{tabular}{lrrr}
			\hline
			reference prior & NRMSE& ${\rm P_{CI}(95\%)} $ & ${\rm L_{CI}(95\%)} $ \\
			\hline
			calibrated computer model  & 18 & / & /  \\
			calibrated computer model and discrepancy  & .28 & .88 & .66   \\
			\hline
		       JR prior & NRMSE & ${\rm P_{CI}(95\%)} $ & ${\rm L_{CI}(95\%)} $ \\
			\hline
			calibrated computer model  & .24 & / & /  \\
			calibrated computer model and discrepancy  & .21 & .98 & .92   \\
			\hline
                         MLE    & NRMSE & ${\rm P_{CI}(95\%)} $ & ${\rm L_{CI}(95\%)} $ \\
			\hline
			calibrated computer model  & .72 & / & /  \\
			calibrated computer model and discrepancy  & .24 & .99 & .80   \\
                        \hline
		\end{tabular}
	\end{table}

We compare the prediction with the reference prior and JR prior  in Figure~\ref{fig:eg_prediction_calibration}. Also included is the prediction with the MLE, in which we first maximize over the mean parameter and variance parameter, and then numerically maximize the profile likelihood of the rest of the parameters. 
First, the prediction by the calibrated computer model with the reference prior clearly overestimates the mean effect, caused by the posterior samples of the large correlation and variance parameter shown in Figure \ref{fig:eg_posterior_calibration}. In comparison, the prediction by the calibrated computer model with the JR prior is  more accurate. The NRMSE is 18 and 0.24, using the prediction by  the calibrated computer model with the reference prior and JR prior, respectively. 

The prediction combining the calibrated computer model and  discrepancy function is graphed as the colored solid curves in Figure \ref{fig:eg_prediction_calibration}.  The model with the JR prior has a lower NRMSE than the one with the reference prior shown in Table~\ref{tab:predicition_calibration}, and as importantly, produces 95\% posterior credible interval covered around 95\%  of held-out points in the target function. In contrast, the model with the reference prior seems overconfident in their accuracy assessment, caused by the large variance of the posterior mean parameter. 

The prediction of the GaSP calibration with the JR prior is also better than the one with the MLE in terms of NRMSE. For the MLE, the uncertainty of the estimated parameters is typically hard to quantify, when the sample size is small. Besides, the likelihood of the calibration parameter normally has multiple local modes, indicating the MLE should be operated with caution.

\section{Concluding remarks}
\label{sec:conclusion}
We have introduced the JR prior for emulation, calibration and   variable selection in UQ. This prior performs as well as the reference prior in emulation, because the marginal posterior mode estimation with the JR prior is robust. The JR prior is considerably faster as the closed form derivative is easy to compute. Furthermore, the marginal posterior mode with the JR prior can identify the inert inputs with no extra computational cost, whereas the marginal posterior mode of the reference prior and other priors may not be both robust in posterior mode estimation and accurate in identifying the inert inputs by the posterior mode. 
 In calibration,  the JR prior is helpful for parameter identification  with the current choice of the prior parameters, which avoids the correlation from being estimated to be too large.  The choice of default prior parameters is still an open problem.  A principle way of determining the prior parameters is needed for the tradeoff in predictive accuracy and identifiability of parameters in calibration. 






 

%
%
%

\bibliographystyle{ba}
\bibliography{References_2018}

\begin{acknowledgement}
The research of Mengyang Gu was part of his PhD thesis at Duke University. The author thanks Jim Berger for his guidance and   insightful discussions.
\end{acknowledgement}

\end{document}